\documentclass[journal]{IEEEtran}
\usepackage{amsmath}
\usepackage{amssymb}
\usepackage{amsthm}
\usepackage{bm}
\usepackage{lipsum}
\usepackage{cuted}	
\usepackage[section]{placeins}
\usepackage[T1]{fontenc}
\usepackage[dvips]{graphicx}
\usepackage{algorithm}
\usepackage{algorithmic}
\usepackage{slashbox}
\usepackage[labelformat=simple]{subcaption}
\usepackage{caption}

\usepackage[numbers,sort,compress]{natbib} 

\usepackage{color}
\usepackage{xcolor}
\usepackage{slashbox}
\usepackage{multirow}
\usepackage{multicol}
\newtheorem{theorem}{Theorem}

\usepackage[normalem]{ulem}

\newtheorem{remark}{Remark}
\usepackage{mathtools}

\begin{document}

\title{Coherent Compensation-Based Sensing for Long-Range Targets in Integrated Sensing and Communication System}
\author{Lin Wang,~\IEEEmembership{Student Member,~IEEE},
        Zhiqing Wei,~\IEEEmembership{Member,~IEEE},
        Xu Chen,~\IEEEmembership{Member,~IEEE},\\
        and Zhiyong Feng,\label{key}~\IEEEmembership{Senior Member,~IEEE}
       
\thanks{

This work was supported in part by the National Natural Science Foundation of China (NSFC) under Grant 92267202 and 62321001, in part by the Fundamental Research Funds for the Central Universities under Grant 2024ZCJH01, and in part by the National Natural Science Foundation
of China (NSFC) under Grant 62271081 and U21B2014.

A part of this paper has been published in the 2023 21st International Symposium on Modeling and Optimization in Mobile, Ad Hoc, and Wireless Networks (WiOpt) \cite{Wiopt}.

Lin Wang, Zhiqing Wei, and Zhiyong Feng are with the Key Laboratory of Universal Wireless Communications, Ministry of Education, School of Information and Communication Engineering, Beijing University of Posts and Telecommunications, Beijing, 100876, China (email: \{wlwl; weizhiqing; fengzy\}@bupt.edu.cn).

Xu Chen is with Innovation Centre of Mobile Communications, China Academy of Information and Telecommunications Technology (CAICT), Beijing, 100191, China (email: chenxu1@caict.ac.cn). 

%

Corresponding authors: Zhiyong Feng, Zhiqing Wei.
}}


\maketitle

\begin{abstract}
Integrated sensing and communication (ISAC) is a promising candidate technology for 6G due to its improvement in spectral efficiency and energy efficiency. Orthogonal frequency division multiplexing (OFDM) signal is a mainstream candidate ISAC waveform. However, there are inter-symbol interference (ISI) and inter-carrier interference (ICI) when the round-trip delay exceeds the cyclic prefix (CP) duration for OFDM signals, which limits the maximum sensing range of ISAC system. When detecting a long-range target, the wide beam inevitably covers the close-range target, of which the echo's power is much larger than that of the long-range target. In order to tackle the above problem, a multiple signal classification (MUSIC) and least squares (LS)-based spatial signal separation method is proposed to separate the echo signals reflected from different targets. Moreover, a coherent compensation-based sensing signal processing method at the receiver is proposed to enhance the signal to interference plus noise power ratio (SINR) of the OFDM block for generating the range-Doppler map (RDM) with higher SINR. Simulation results reveal that the proposed method greatly enhances the SINR of RDM by 10 dB for a target at 500 m compared with two-dimensional fast Fourier transform (2D-FFT) method. Besides, the detection probability is also significantly improved compared to the benchmarking method.
\end{abstract}

\begin{keywords}
Integrated sensing and communication, multiple input multiple output, orthogonal frequency division multiplexing, multiple signal classification, long-range targets, coherent compensation.
\end{keywords}

\IEEEpeerreviewmaketitle
\section{Introduction}
Wireless communication and sensing were developed separately over the past few decades. The emerging machine-type applications in the 5th generation advanced (5G-A) networks and 6th generation (6G) networks have introduced a demand for dual wireless communication and wireless sensing functions in the same network \cite{wei iot}\cite{liufan isac survey}\cite{andrew zhang survey}. As a promising candidate technology for future networks, integrated sensing and communication (ISAC) aims to realize wireless communication and wireless sensing using the same hardware and physical resource, improving spectral and hardware utilization efficiency \cite{feng china communication}\cite{wei network}\cite{wangyuan}.

ISAC waveform has a great impact on the performance of sensing and communication in ISAC system. Due to the advantages such as high spectrum efficiency and capability against frequency-selective fading, orthogonal frequency division multiplexing (OFDM) waveform is a key-enabled technology for communication \cite{zhou survey}. Besides, the excellent sensing properties such as thumbtack-like ambiguity function shape make OFDM an important candidate ISAC waveform \cite{OFDM AF}.  Thus, some sensing signal processing methods based on OFDM waveform were proposed. Sturm \textit{et al.} \cite{Sturm 1}  proposed a two-dimensional fast Fourier transform (2D-FFT)-based OFDM radar signal processing method in the modulation symbol domain. Compared with the classical radar correlation method in time domain, the FFT-based sensing method offers a very high dynamic range and low computational complexity
 \cite{Sturm 2}. Chen \textit{et al.} \cite{chenxu twc} proposed a multiple signal classification (MUSIC)-based signal processing method that realizes the super-resolution for angle, range, and velocity estimation. Wei \textit{et al.} \cite{quhanyang} proposed an iterative signal processing method to improve the sensing accuracy with low complexity.

For long-range sensing using OFDM waveform, the maximum sensing range is limited by the cyclic prefix (CP) duration. The intrinsic CP in the OFDM signal is designed for communication. The CP duration must be greater than the maximum delay spread. Similarly, the maximum round-trip delay for sensing targets should be smaller than CP duration for the monostatic sensing of ISAC base station (BS) \cite{wuyongzhi cp}. Otherwise, there are inter-symbol interference (ISI) and inter-carrier interference (ICI) for the round-trip delay exceeding the CP duration, which limits the maximum sensing range \cite{wuyongzhi cp2}. Hakobyan \textit{et al.} \cite{repeated OFDM} proposed to transmit repetitive OFDM symbols, which does not cause ISI and ICI but limits the data rate. Yuan \textit{et al.} \cite{yuan multiple CP} proposed to prepend an extra sensing CP to the transmitted signal, decoupling the sensing and communication requirements. However, this requires combining multiple OFDM symbols for large FFT numbers which is difficult to implement in hardware. Wu \textit{et al.} \cite{wu VCP} flexibly re-segmented the transmitted and received signals into sub-blocks at the receiver according to the sensing requirements. A virtual CP (VCP) is added to received sub-blocks to form a cyclic shift signal. However, the re-segmented sub-blocks disrupt the OFDM structure, leading to amplified noise in the element-wise-division process. For IEEE 802.11ad-based ISAC signal, Tang \textit{et al.} \cite{Tang}\cite{Tang 2} utilized the pilot signal and a few interference-free OFDM symbols to achieve long-range sensing. This method is valid when the round-trip delay is greater than one OFDM duration, which is not considered under existing cellular network deployment. 
Tang \textit{et al.} \cite{Tang SCP} proposed to alternately deploy zero-power reference signal (ZP-RS) and non-zero-power reference signal (NZP-RS) at subcarriers, which extends the round-trip delay for echo signal without ISI to CP duration plus half of OFDM symbol duration. This novel RS design holds significant potential in scenarios with abundant spectrum resources.
We proposed a coherent compensation-based sensing processing method without changing the existing signal structure to enhance the SINR of OFDM block, thus improving the sensing performance \cite{Wiopt}. However, the previous version does not consider the multiple input multiple output (MIMO) and multiple targets.

Generally, ISAC can be implemented by a unified dual-functional waveform or dedicated communication and sensing waveform \cite{feng china communication}. The unified dual-functional waveforms require sensing targets and communication users in the same direction \cite{comm and sensing in the same direction}. Moreover, the communication and sensing requirements are sometimes asymmetrical, that is, sensing services do not always occur in all communication time slots, and they are related to the targets’ motion state \cite{comm and sensing requirements asymmertrical 1}\cite{comm and sensing requirements asymmertrical 2}. In \cite{comm and sensing requirements asymmertrical 1}, the ISAC system with different sensing frequencies and data frame rates is considered. The dedicated communication and sensing waveforms occupy the orthogonal time domain \cite{time division}, frequency domain \cite{frequency division}, or beam domain resources \cite{Andrew zhang multibeam} to avoid mutual interference. However, whether the unified dual-functional waveform or dedicated communication and sensing waveform is used, MIMO-BS covers a large area due to the wide beam when sensing long-range targets  \cite{wide beam}.
Therefore, the BS simultaneously receives the echo signals from long-range targets and close-range targets, which can't be distinguished in the spatial domain using beam-based direction of arrival (DoA) estimation. The interference caused by the close-range targets may drown the echo signals of the long-range targets, especially when the ISI and ICI also exist in the echo signals of the close-range target. Most existing researches focus on the interference between communication beams and sensing beams \cite{interference 1}\cite{interference 2} and rarely considers intra-beam interference, which is also important for long-range sensing. 

In this paper, we consider both wide-beam and long-range sensing of BS monostatic active sensing. First, we proved that the DoA estimation is not prominently affected by ISI and ICI. The multiple signal classification (MUSIC) method is exploited to estimate the angle of targets within the same beam. Then, a least squares (LS)-based estimation is proposed to achieve the separation of echo signals reflected from different targets in the spatial domain. Moreover, the ISI and ICI of different echo signals are also separated. For long-range targets, a coherent compensation-based sensing signal processing method is proposed to improve the signal to interference plus noise power ratio (SINR) of range-Doppler map (RDM). The main contributions of this paper are summarized as follows.
 \begin{itemize}
	\item [1)] 
	To mitigate the impact of ISI and ICI on sensing, the MUSIC method is exploited to estimate the angles of multiple targets within the same beam. Then, an LS-based signal separation method is proposed to achieve the separation of different echo signals in the spatial domain.
	
	\item [2)]
	We model the ISI and ICI of echo signals reflected from long-range targets and propose a coherent compensation-based sensing signal processing method. The samples after each received OFDM symbol are added to the head of that OFDM symbol to achieve coherent compensation, maintaining the original OFDM structure. The SINR of each OFDM block and SINR of RDM are enhanced, thereby improving the sensing performance.
	
	\item [3)]
	We derive the theoretical values of the SINR of OFDM blocks obtained by the coherent compensation-based method and prove that there exists an optimal number of compensation samples. In addition, the theoretical values of the SINR of RDM obtained by the proposed method and benchmarking method are derived, which match well with the simulation results.
	
	\item [4)]
	Extensive simulations have been conducted to validate both the proposed signal separation method in the spatial domain and coherent compensation-based long-range sensing method. The simulation results demonstrate that the proposed method enhances the SINR of RDM by around $10$ dB compared with the benchmarking method for the task of sensing the target at $500$ m.  The detection probability is also simulated to prove the validity of the proposed method for long-range sensing. 
\end{itemize}

Compared with the earlier conference version \cite{Wiopt}, this paper expands its scope to include multiple antennas and multiple targets. A novel LS estimation-based signal separation method is proposed to achieve signal separation. The difference in Doppler shift between the OFDM symbol and compensation samples is considered to derive the theoretical values for the SINR of the OFDM block. Besides, more performance metrics such as theoretical SINR of RDM and detection probability are analyzed to validate the proposed method.
 
The subsequent sections of this paper are structured as follows. Section \ref{sec: System Model} presents the system model of this paper. The signal separation method and coherent compensation method are introduced in Section \ref{sec: Signal Separation and Coherent Compensation Algorithm}. Section \ref{sec: Performance Analysis} analyzes the theoretical values for the SINR of OFDM block, SINR of RDM, and the maximum sensing range supported by the proposed method. Extensive simulations are conducted to validate the proposed method in Section \ref{sec: Simulation Results and Analysis}, while Section \ref{sec: Conclusion} serves as the the conclusion of the paper.

{\bf{Notations}}: The lowercase letters, lowercase bold letters, and uppercase bold letters represent the scalars, column vectors, and matrices, respectively; $(\cdot)^T$, $(\cdot)^*$, and $(\cdot)^H$ represent the transpose, conjugate, and Hermitian operators, respectively; $\mathbb{E} \{\cdot\}$ and $\mathbb{V} \{\cdot\}$ represent the expectation and variance operators, respectively; $\left| \cdot \right|$,  ${\left\| {\cdot} \right\|_\ell }$, $\text{eig}(\cdot)$, and $tr\left( \cdot \right)$  represent the modulo, $\ell$-norm,  eigenvalue decomposition, and trace operators, respectively; $\left[\cdot \right]^{-1}$ and ${\left[ \cdot \right]^\dag }$  represent the inverse and Moore-Penrose inverse of a matrix with size $N \times N$ and $M \times N$, respectively; $\lfloor \cdot \rceil$ rounds toward the nearest integer,    $\mathbb{C}^{M \times N}$ represents the set of $M \times N$ complex numbers, and $ {\cal C}{\cal N}( {a,{\sigma ^2}} )$ represents the circularly
symmetric complex Gaussian (CSCG) random variable with mean $a$ and variance ${\sigma ^2}$.

\section{System Model}
\label{sec: System Model}
In this section, the BS monostatic sensing model and antenna array model are presented. Subsequently, the transmitted and received signals are characterized as the prerequisite for further signal separation and coherent compensation methods.
\begin{figure}[ht]
	\centering
	\includegraphics[scale=0.56]{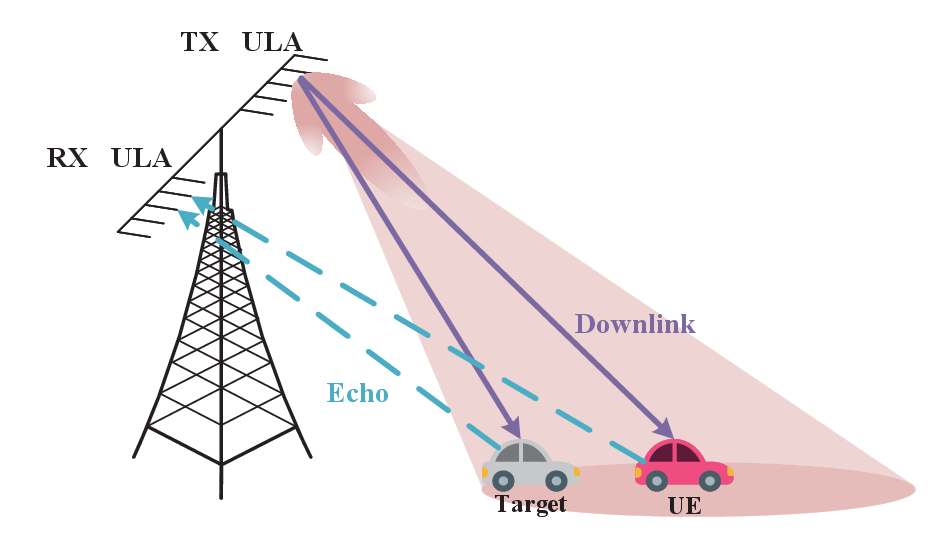}
	\caption{\centering {BS monostatic sensing model.}}
	\label{BS downlink sensing model}
\end{figure}
\subsection{ BS monostatic sensing model}
A MIMO ISAC BS equipped with spatially well-separated transmit and receive antenna arrays is considered to achieve full-duplex operation for conducting monostatic sensing \cite{FD}, as depicted in Fig. \ref{BS downlink sensing model}. The BS transmits OFDM signals to communicate with user equipment (UE) while simultaneously receiving echo signals for sensing the surrounding targets of UE  \footnote{ This paper takes the unified dual-functional waveform as an example. However, the proposed sensing method applies to other ISAC systems such as time-division, frequency-division, etc., as long as the OFDM waveform is used for sensing. As we focus on sensing signal processing, only the moment when communication and sensing exist simultaneously are considered. Note that both functionalities may not always coincide all the time in the practical system \cite{comm and sensing requirements asymmertrical 1}.}. Additionally, perfect beam alignment with the target UE is assumed. Note that BS can receive echo signals reflected from other sensing targets within the beam coverage area due to the wide beam.

\subsection{Antenna array model}
In this paper, the uniform linear arrays (ULAs) are used, as depicted in Fig. \ref{figure_2_a}.  The transmit ULA and receive ULA have $N_t$ and $N_r$ antenna elements, respectively, with element spacing denoted by $d_t$ and $d_r$, respectively.  Then, the steering vectors of transmit ULA and receive ULA are respectively given by
\begin{equation}
{\bf a}\left( \theta_t  \right) = {\left[ {1,{e^{-j\frac{{2\pi }}{\lambda }{d_t}\sin \theta _t}}, \cdots ,{e^{-j\frac{{2\pi }}{\lambda }\left( {{N_t} - 1} \right){d_t}\sin \theta _t}}} \right]^T},
\label{transmit steering vector}
\end{equation}
\begin{equation}
{\bf b}\left( \theta_r  \right) = {\left[ {1,{e^{-j\frac{{2\pi }}{\lambda }{d_r}\sin {\theta _r}}}, \cdots ,{e^{-j\frac{{2\pi }}{\lambda }\left( {{N_r} - 1} \right){d_r}\sin {\theta _r}}}} \right]^T},
\label{receive steering vector}
\end{equation}
where $\lambda$ is the wavelength, $\theta _t$ and $\theta _r$ are the elevation angles of the transmitted and received signals, respectively. Typically, these angles are assumed to be equal for far-field monostatic sensing,  i.e., $\theta_t=\theta_r=\theta$.

Note that although the ULAs are considered, the insights from this paper can be easily extended to the uniform rectangular arrays (URA). Since the URA consists of two one-dimensional ULAs, the beam patterns of URA and ULA are the same on elevation cut or azimuth cut.  Fig. \ref{figure_2_a} and Fig. \ref{figure_2_b} depict the vertical ULA and URA, respectively. Thus, the half power beamwidth (HPBW) is the same at elevation cut, but different at azimuth cut, as shown in Fig. \ref{figure_2_c}-Fig. \ref{figure_2_f}.

\begin{figure*}[!t]
	\setcounter{subfigure}{0}
	\begin{subfigure}[b]{0.32\textwidth}
		\includegraphics[width=6.5cm]{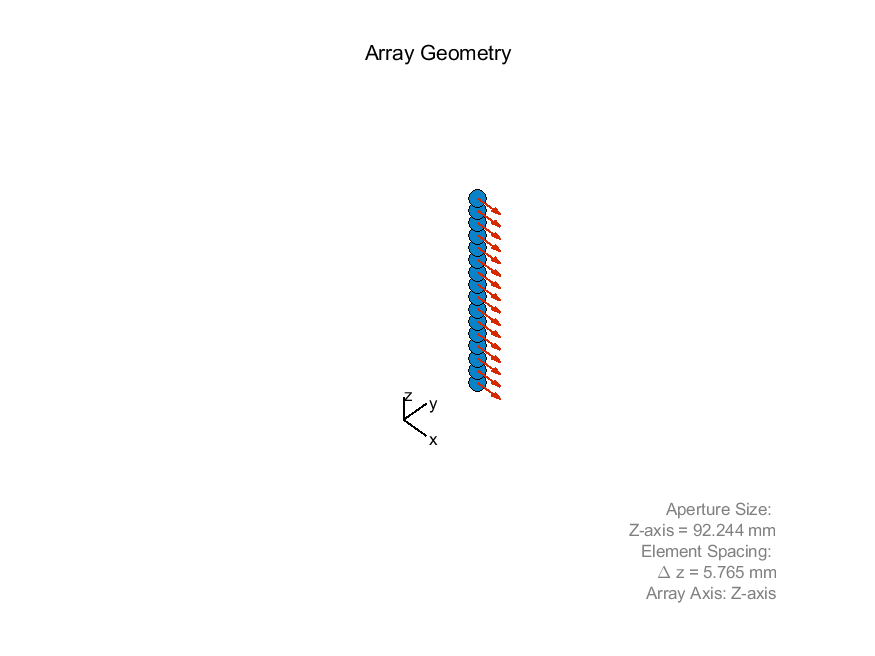}
		\caption{{ULA geometry}}
		\label{figure_2_a}
	\end{subfigure}%
	\setcounter{subfigure}{2}
	\begin{subfigure}[b]{0.32\textwidth}
		\includegraphics[width=6.5cm]{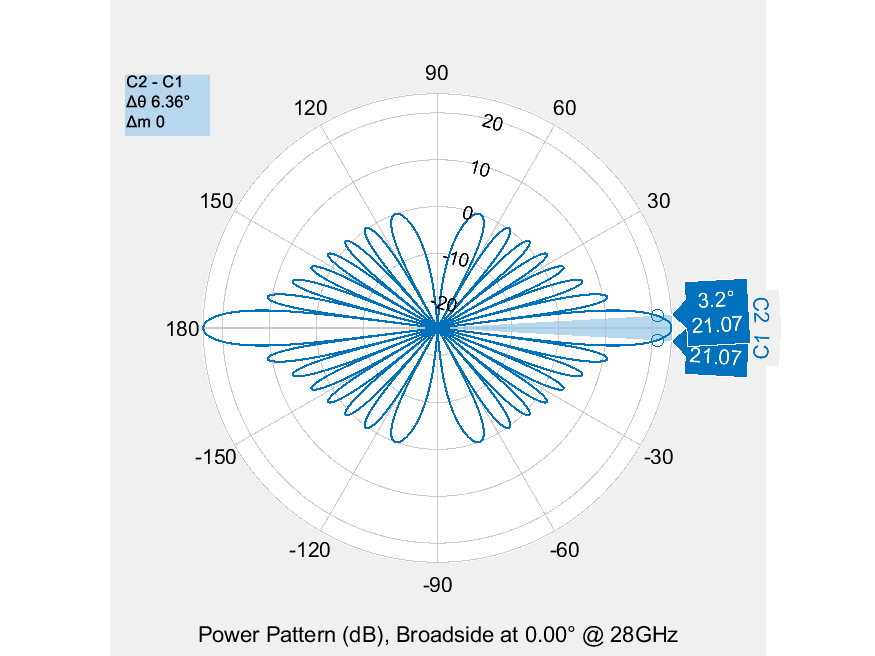}
		\caption{{HPBW at elevation cut of ULA}}
		\label{figure_2_c}
	\end{subfigure}%
	\setcounter{subfigure}{4}
	\begin{subfigure}[b]{0.32\textwidth}
		\includegraphics[width=6.5cm]{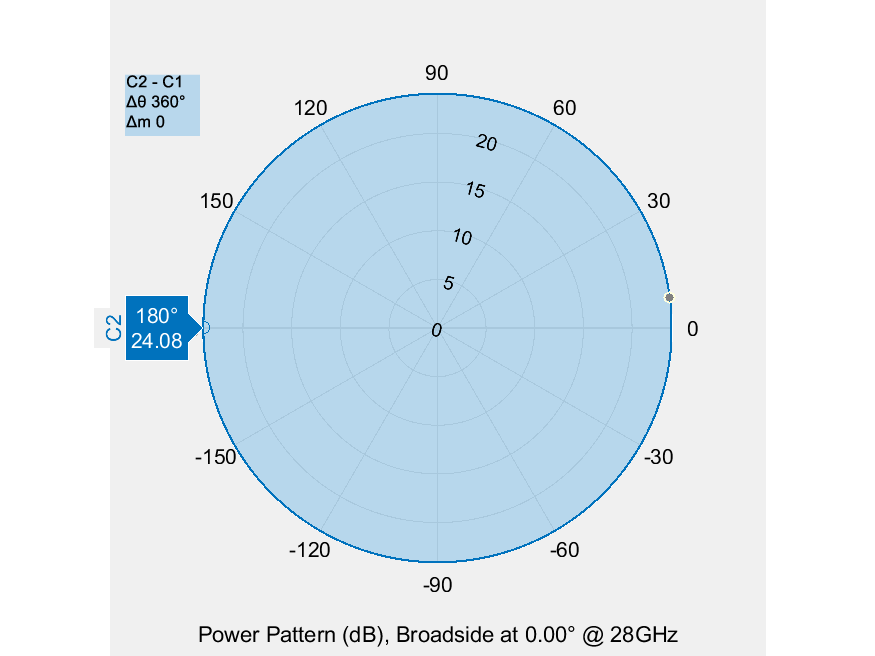}
		\caption{{HPBW at azimuth cut of ULA}}
		\label{figure_2_e}
	\end{subfigure}%
	\\
	\setcounter{subfigure}{1}
	\begin{subfigure}[b]{0.32\textwidth}
		\includegraphics[width=6.5cm]{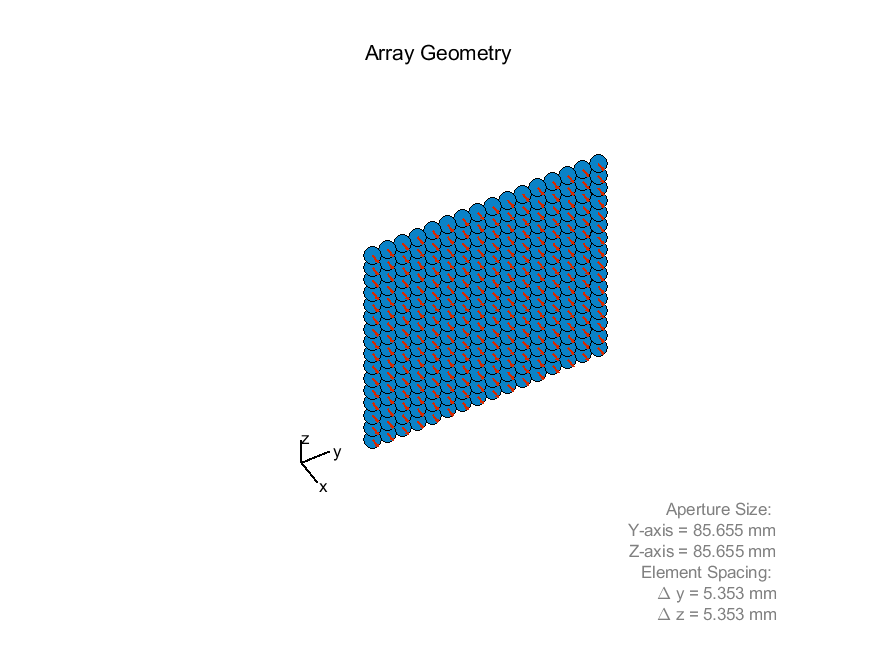}
		\caption{{URA geometry}}
		\label{figure_2_b}
	\end{subfigure}%
	\setcounter{subfigure}{3}
	\begin{subfigure}[b]{0.32\textwidth}
		\includegraphics[width=6.5cm]{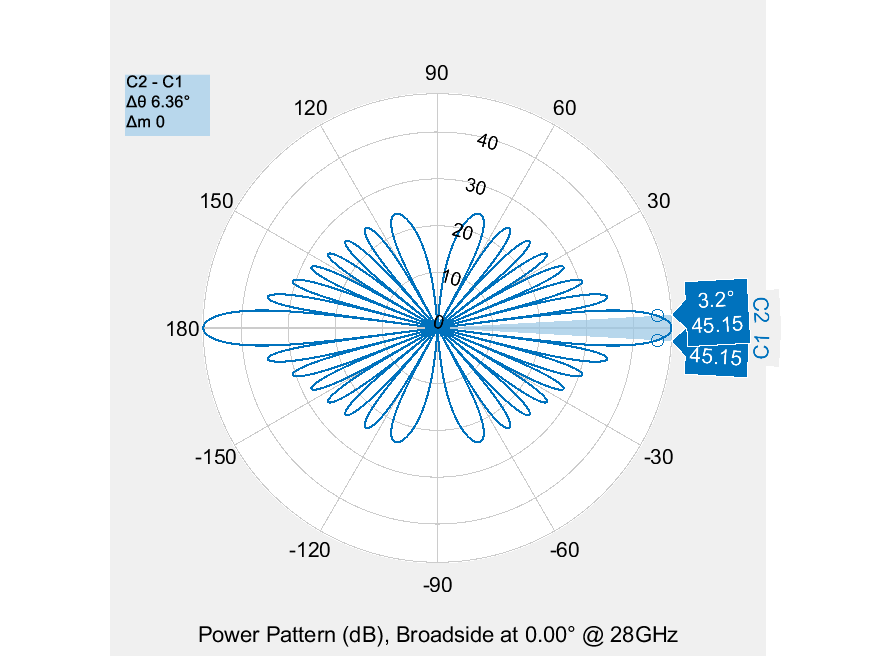}
		\caption{{HPBW at elevation cut of URA}}
		\label{figure_2_d}
	\end{subfigure}%
	\setcounter{subfigure}{5}
	\begin{subfigure}[b]{0.32\textwidth}
		\includegraphics[width=6.5cm]{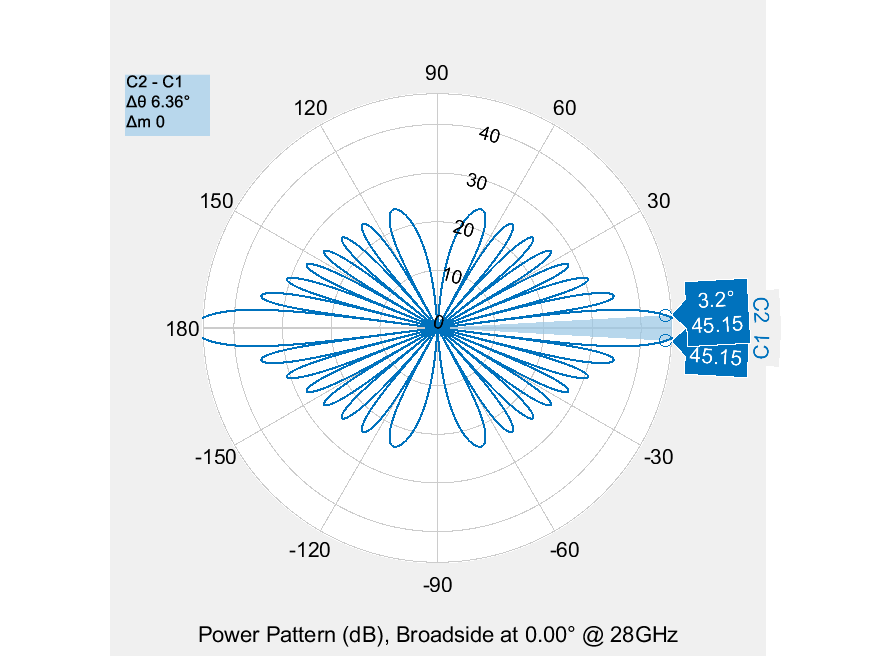}
		\caption{{HPBW at azimuth cut of URA}}
		\label{figure_2_f}
	\end{subfigure}%
	\caption{ {The comparison of HPBW of $16 \times 1$ ULA and $16 \times 16$ URA at elevation and azimuth cuts.}}
	\label{figure_2}
\end{figure*}
\subsection{Transmitted signal model}
In this paper, BS transmits OFDM signals to UEs. The baseband OFDM signal is \cite{chenxu twc}
\begin{equation}
	s\left( t \right) = \frac{1}{{\sqrt {{N_c}} }}\sum\limits_{m = 0}^{M - 1} {\sum\limits_{k = 0}^{{N_c} - 1} {S_{m,k}^{}{e^{j2\pi k\Delta f\left( {t - mT} \right)}}} u\left( {t - mT} \right)} ,
\end{equation}
where $N_c$ and $M$ represent the numbers of subcarriers and symbols, respectively; $S_{m,k}$ is the modulation symbol transmitted at the $k$th subcarrier of the $m$th OFDM symbol;  $\Delta f$ is the subcarrier spacing,  $T=T_{cp}+T_d$ is the total duration of the OFDM symbol block, $T_{cp}$ and $T_d$ are the CP duration and the OFDM symbol duration, respectively. Moreover, $u(t)$ is the rectangular function, which is given by
\begin{equation}
	u\left( t \right) = \left\{ {\begin{array}{*{20}{c}}
			1,&{ - {T_{cp}} \le t \le {T_d}}\\
			0,&{\text{otherwise}}
	\end{array}}. \right.
	\label{pulse shaping function}
\end{equation}
Then, the OFDM signal transmitted by antenna array is
\begin{equation}
	{\bf{x}_f}\left( t \right)={\bf{x}}\left( t \right){e^{j2\pi {f_c}t}} = {{\bf{w}}_{tx}}s\left( t \right){e^{j2\pi {f_c}t}},
\end{equation}
where ${\bf{x}}\left( t \right)={{\bf{w}}_{tx}}s\left( t \right)$ is the baseband OFDM signal transmitted at different antennas, $f_c$ is the carrier frequency, ${\bf w}_{tx} \in \mathbb{C}^{N_t \times 1}$ is the precoding vector at the transmitter. Specifically, when the transmit beam is aligned with the target UE, ${\bf w}_{tx}$ can be obtained by using the LS beamforming method, which can be expressed as \cite{Andrew zhang multibeam}
\begin{equation}
	{{\bf w}_{tx}}{\rm{ = }}{\left[ {{{\bf a}^T}\left( \theta_t  \right)} \right]^\dag }.	
	\label{transmit beam LS method}
\end{equation}

\subsection{Received echo signal model}
The baseband echo signals received by BS can be given by
\begin{equation}
	{\bf y}\left( t \right) = \sum\limits_{u = 0}^{\mathclap{U - 1}} {{{\tilde \alpha }_u}{\bf b}\left( {{\theta _u}} \right){{\bf a}^T}\left( {{\theta _u}} \right){\bf x}\left( {t - {\tau _u}} \right)} {e^{j2\pi {f_{d,u}}t}} + {\bf n}\left( t \right),
	\label{received signal in time domain}
\end{equation}
where ${\bf b}(\theta_u)$ and ${\bf a}(\theta_u)$ are the receive and transmit steering vectors for the $u$th target, respectively; ${\theta _u}$ is the elevation angle between the $u$th target and BS, $U$ is the number of targets within the coverage area of the transmit beam, which is not greater than the size of receive ULA,  i.e., $U \le N_r$; ${\tau _u}{\rm{ = }}2{d_u}/{c_0}$ and ${f_{d,u}}{\rm{ = }}2{v_u}{f_c}/{c_0}$ are the round-trip delay and Doppler frequency shift between the $u$th target and BS, respectively, with $c_0$, ${d_u}$, and ${v_u}$ being the light speed, range, and radial velocity of the $u$th target, respectively; ${\tilde \alpha _u}{\rm{ = }}\sqrt {{P_t}{G_t}{G_r}/{N_t}} {\alpha _u}{\beta _u}{e^{-j2\pi {f_c}{\tau _u}}}$ is the attenuation factor between the $u$th target and BS, with ${\alpha _u}$ being the radar cross section (RCS); $P_t$ is the transmit power, $G_t$ and $G_r$ are the transmit and receive array gains, respectively; ${\beta _u} = \sqrt {\frac{{{\lambda ^2}}}{{{{\left( {4\pi } \right)}^3}d_u^4}}} $ is the free space path loss between the $u$th target and BS. Moreover, ${\bf n}(t)$ is the noise following a CSCG stochastic process. In active sensing, the BS contains priori information about the transmitted signal such as payload data and angle. Thus, the received signal can be expressed as
\begin{align}
	y\left( t \right) 
	&= \sum\limits_{u = 0}^{\mathclap{U - 1}} {{{\tilde \alpha }_u}{\bf w}_{rx}^T{\bf b}\left( {{\theta _u}} \right){{\bf a}^T}\left( {{\theta _u}} \right){\bf x}\left( {t - {\tau _u}} \right)} {e^{j2\pi {f_{d,u}}t}} + \tilde n\left( t \right),
	\label{traditional received signal}
\end{align}
where ${\bf w}_{rx} \in \mathbb{C}^{N_r \times 1}$ is the receive precoding vector, $\tilde n\left( t \right)={\bf w}_{rx}^T{\bf n}\left( t \right)$ is the transformed noise. Generally, the maximum gain is achieved by aligning the receive beam with the transmit beam, i.e., ${{\bf w}_{rx}}{\rm{ = }}{[ {{{\bf b}^T}\left( \theta_0  \right)} ]^\dag }$, with $\theta_0$ being the elevation angle of UE. 

Nevertheless, owing to the wide beam, the coverage area encompasses both long-range and close-range targets.  The close-range targets may deteriorates the sensing performance of the long-range targets, especially when the round-trip delay of the close-range target also exceeds the CP duration, which leads to the ISI and ICI with high power. To tackle this challenge, we further propose a MUSIC and LS estimation-based spatial signal separation method and a coherent compensation-based long-range sensing method in the next section to improve the sensing performance.

\section{Signal Separation and Coherent Compensation Methods}
\label{sec: Signal Separation and Coherent Compensation Algorithm}
As the DoA estimation is not prominently affected by ISI and ICI, we first adopt MUSIC method to distinguish targets with different angles within the same beam. Additionally, an LS estimation-based signal separation method is proposed to separate echo signals in the spatial domain effectively. Furthermore, we propose a coherent compensation-based method for long-range sensing, aiming to enhance the SINR of RDM.

\subsection{MUSIC-based LS signal separation method}
For long-range targets, their round-trip delays exceed the CP duration. Thus, there are ISI and ICI in the frequency domain, which severely deteriorates the detection probability and the accuracy of range and velocity estimation based on the 2D-FFT method. However, the DoA estimation is not prominently affected by ISI and ICI. As shown in \eqref{received signal in matrix}, ${{\bf x}_r}\left( t \right)$ can represent both the useful echo signal and ISI, whose array manifold matrix is the same. Consequently, ISI can be used to assist DoA estimation. 

In this paper, multiple targets fall into the same beam, leading to a challenge in estimating the angles  using traditional beam-based DoA estimation method according to the Rayleigh limit \cite{rayleigh limit}. On the other hand, the super-resolution methods such as MUSIC method \cite{MUSIC}, estimating signal parameter via rotational invariance techniques (ESPRIT) method \cite{ESPRIT} and minimum variance distortionless response (MVDR) method \cite{MVDR} can be utilized to estimate the DoA. Given that the signals are uncorrelated, the MUSIC method theoretically has infinite resolution and can achieve unbiased estimation of angles \cite{MUSIC resolution}. Therefore, the MUSIC method is employed for DoA estimation, and \eqref{received signal in time domain} is recast in a compact form as
\begin{equation}
	{\bf y}\left( t \right){\rm{ = }}{\bf B}\left( {\boldsymbol{\theta}}  \right){{\bf x}_r}\left( t \right) + {\bf n}\left( t \right),
	\label{received signal in matrix}
\end{equation}
where 
\begin{equation}
	{\bf B}\left( \boldsymbol{\theta}  \right){\rm{ = }}\left[ {{\bf b}\left( {{\theta _0}} \right),{\bf b}\left( {{\theta _1}} \right), \cdots ,{\bf b}\left( {{\theta _{U - 1}}} \right)} \right]
\end{equation}
is the array manifold matrix with $\boldsymbol{\theta}=\left[\theta_0, \theta_1,\cdots,\theta_{U-1}\right]^T$ being the DoA vector, and
\begin{equation}
	{{\bf x}_r}\left( t \right) = \left[ {\begin{array}{*{20}{c}}
			{{{\tilde \alpha }_0}{{\bf a}^T}\left( {{\theta _0}} \right){\bf x}\left( {t - {\tau _0}} \right){e^{j2\pi {f_{d,0}}t}}}\\
			{{{\tilde \alpha }_1}{{\bf a}^T}\left( {{\theta _1}} \right){\bf x}\left( {t - {\tau _1}} \right){e^{j2\pi {f_{d,1}}t}}}\\
			\vdots \\
			{{{\tilde \alpha }_{U - 1}}{{\bf a}^T}\left( {{\theta _{U - 1}}} \right){\bf x}\left( {t - {\tau _{U - 1}}} \right){e^{j2\pi f_{d,U-1}t}}}
	\end{array}} \right]
\label{receive echo signal in compact form}
\end{equation}
is the echo signals reflected from different targets. The $N$ snapshots of ${\bf y}(t)$ can be stacked as ${\bf Y}\in\mathbb{C}^{N_r \times N}$. Furthermore, the correlation matrix of ${\bf y}(t)$ can be estimated as \cite{MUSIC}
\begin{equation}
	{\bf \hat R} = \frac{1}{N}{\bf Y}{{\bf Y}^H},
\end{equation}
which is the maximum likelihood estimation. By conducting eigenvalue decomposition on ${\bf \hat R}$, we can obtain  \cite{MUSIC}
\begin{equation}
	\left[ {{\bf A},\boldsymbol{\Sigma} } \right] = \text{eig}( {\bf \hat R} ),
\end{equation}
where ${\bf A}$ and $\boldsymbol{\Sigma}$ are the eigenvector and eigenvalue matrices, respectively. Subsequently, we gain \cite{MUSIC}
\begin{equation}
	{\bf \hat R}{\rm{ = }}{{\bf A}_S}{{\bf \Sigma} _S}{\bf A}_S^H + {{\bf A}_N}{{\bf \Sigma} _N}{\bf A}_N^H,
\end{equation}
where ${\bf \Sigma} _S$ and ${\bf \Sigma} _N$ are the diagonal matrices composed of large and small eigenvalues, respectively; ${\bf A} _S$ and ${\bf A} _N$ are the signal and noise subspaces, respectively. Since the noise subspace ${\bf A} _N$ and array manifold subspace of ${\bf B}\left( {\boldsymbol{\theta}}\right)$ are orthogonal, the DoAs can be estimated using the following angle spectral estimation.
\begin{equation}
	 \mathop {\arg {\rm{ }}\min }\limits_\theta  {\rm{  }}{{\bf b}^H}\left( \theta  \right){{\bf A}_N}{\bf A}_N^H{\bf b}\left( \theta  \right).
	 \label{spectral estimation}
\end{equation}
Moreover, \eqref{spectral estimation} is equivalent to searching for the peaks of
\begin{equation}
{P_{MUSIC}} = \frac{1}{{{\rm{ }}{{\bf b}^H}\left( \theta  \right){{\bf A}_N}{\bf A}_N^H{\bf b}\left( \theta  \right)}}.
\label{MUSIC}
\end{equation}

For monostatic sensing, the prior information contained in the transmitted signal, such as the angle and the width of transmit beam, is known. Therefore, the search range in \eqref{MUSIC} can be significantly narrowed down. We assume that the estimated angles and the number of estimated targets approximate the actual situation. In contrast to blind signal separation, the structure of ${\bf B}\left( {\boldsymbol{\theta}}\right)$ exhibiting a Vandermonde structure is known. With the estimated angle, ${\bf B}\left( {\boldsymbol{\theta}}\right)$ becomes known. 
However, the target's position and speed are random, and their distributions are unknown, making it difficult to obtain the statistical characteristics of \eqref{receive echo signal in compact form}. Therefore, the LS estimation method is utilized to achieve the separation of echo signals reflected from different targets in the spatial domain, as shown in Theorem \ref{theorem 1}.


\begin{theorem}
\label{theorem 1}
Applying LS estimation method to separate echo signals from \eqref{received signal in matrix}, the separated signals can be obtained as
\begin{equation}
	{\bf \bar y}\left( t \right) = {\left[ {{\bf B}\left( \boldsymbol{\theta}  \right)} \right]^\dag }{\bf y}\left( t \right){\rm{ = }}{{\bf x}_r}\left( t \right) + {\bf \bar n}\left( t \right),
	\label{signal seperation}
\end{equation}
where ${\left[ {{\bf B}\left( {\boldsymbol{\theta}}  \right)} \right]^\dag }{\rm{ = }}{\left[ {{{\bf B}^H}\left( {\boldsymbol{\theta}}  \right){\bf B}\left( {\boldsymbol{\theta}}  \right)} \right]^{ - 1}}{{\bf B}^H}\left( {\boldsymbol{\theta}}  \right)$ is the Moore-Penrose inverse of ${{\bf B}\left( \boldsymbol{\theta}  \right)}$, ${\bf \bar n}\left( t \right) = {\left[ {{\bf B}\left( \boldsymbol{\theta}  \right)} \right]^\dag }{\bf n}\left( t \right)$ is the noise after transformation.
\end{theorem}
\begin{proof}
    The proof is provided in {\bf Appendix A}. 
\end{proof}

Then, the separated echo signal reflected from the $u$th target is the $u$th elements in \eqref{signal seperation}, which can be expressed as
\begin{align}
	{{\bar y}^u}\left( t \right) {\rm{ = }}{{\tilde \alpha }_u}{{\bf a}^T}\left( {{\theta _u}} \right){{\bf w}_{tx}}s\left( {t - {\tau _u}} \right){e^{j2\pi {f_{d,u}}t}} + {{\bar n}_u}\left( t \right).
	\label{seperated signal}
\end{align}
where ${\bar n_u}\left( t \right)$ is the transformed noise random process after the LS estimation-based signal separation process. 
    
\begin{theorem}
\label{theorem 2}
     ${\bar n_u}\left( {{t_i}} \right)$ is the sampling of ${\bar n_u}\left( t \right)$ at $t=t_i$, which follows an independent and identically distributed (i.i.d) CSCG distribution, i.e., ${\bar n_u}\left( {{t_i}} \right) \sim {\cal C}{\cal N}\left( {0,{\lambda _u}{\sigma ^2}} \right)$, where $\lambda_u$ is the $u$th diagonal element of ${\left[ {{\bf B}\left( \boldsymbol{\theta}  \right)} \right]^\dag }{[ {{{[ {{\bf B}\left(  \boldsymbol{\theta}  \right)} ]}^\dag }} ]^H}$, and ${\sigma ^2}$ is the power of noise.
\end{theorem}
\begin{proof}
    The proof is located in {\bf Appendix B}.
\end{proof}
According to \eqref{seperated signal}, the recovered frequency domain symbols in the receiver can be derived under different circumstances of round-trip delays.
\subsubsection{Round-trip delay exceeding the CP duration}
Since the ISAC transceivers are co-located, their clocks are synchronized. As shown in Fig. \ref{echo signal model}, when the round-trip delay is greater than the CP duration but smaller than OFDM symbol duration, i.e., $T_{cp}<\tau _u<T$,  the samplings (removing CP) in the $n$th interval can be expressed as
\begin{align}
\bar y_n^u&\left( i \right) = \bar y^u\left( {nT + i{T_s}} \right)\notag\\
 &= \frac{1}{{\sqrt {{N_c}} }}{{\tilde \alpha }_u}{{\bf{a}}^T}\left( {{\theta _u}} \right){{\bf{w}}_{tx}}\label{sampling of receive signal}\\
 &  \cdot \sum\limits_{m = 0}^{M - 1} {\sum\limits_{k = 0}^{{N_c} - 1} {\left[ \begin{array}{l}
S_{m,k}^{}{e^{j2\pi k\Delta f\left( {nT + i{T_s} - mT - {\tau _u}} \right)}}\\
 \cdot {e^{j2\pi {f_{d,u}}\left( {nT + i{T_s}} \right)}}u\left( {nT + i{T_s} - mT - {\tau _u}} \right)
\end{array} \right]} } \notag \\
&{\rm{ + }}{{\bar n}_u}\left( i \right),\notag
\end{align}
where $i=0,1,\cdots,N_c-1$, $T_s=1/B$ is the sampling interval, and ${\bar n_u}\left( i \right)\sim {\cal C}{\cal N}\left( {0,{\lambda _u}{\sigma ^2}} \right)$ is the sampling of ${\bar n_u}\left( t \right)$. 
According to \eqref{pulse shaping function}, the condition $ - {T_{cp}} \le (n-m)T + i{T_s} - {\tau _u} \le {T_d}$ needs to be satisfied. Consequently, $n + \frac{{i{T_s} - {\tau _u} - {T_d}}}{T} \le m \le n + \frac{{i{T_s} - {\tau _u} + {T_{cp}}}}{T}$ is obtained. Given that ${T_{cp}} < {\tau _u} < T$, $m = n - 1, n$  is obtained. 
Generally, the channel response within an OFDM symbol is time-invariant when the subcarrier spacing is much larger than the maximum Doppler shift, i.e., $\Delta f \gg {f_{d,max}}$ (e.g.  $\Delta f > 10{f_{D,\max }}$ as proposed in \cite{subcarrier spacing 10 times larger than doppler}). Thus, the phase rotation caused by the Doppler shift can be regarded as the same for fast-time samples within one OFDM symbol \cite{repeated OFDM} and \eqref{sampling of receive signal} can be recast as
\begin{align}
&y_n^u\left( i \right)= {{\bar n}_u}\left( i \right) \label{sampling of receive signal rewritten}\\
&+ \frac{1}{{\sqrt {{N_c}} }}{{\tilde \alpha }_u}{{\bf{a}}^T}\left( {{\theta _u}} \right){{\bf{w}}_{tx}}\sum\limits_{k = 0}^{\mathclap{{N_c} - 1}} {\left[ \begin{array}{l}
S_{n - 1,k}^{}{e^{j2\pi k\Delta f\left( {T + i{T_s} - {\tau _u}} \right)}}\\
 \cdot {e^{j2\pi {f_{d,u}}nT}}u\left( {T + i{T_s} - {\tau _u}} \right)
\end{array} \right]} \notag\\
 &+ \frac{1}{{\sqrt {{N_c}} }}{{\tilde \alpha }_u}{{\bf{a}}^T}\left( {{\theta _u}} \right){{\bf{w}}_{tx}}\sum\limits_{k = 0}^{{N_c} - 1} {\left[ \begin{array}{l}
S_{n,k}^{}{e^{j2\pi k\Delta f\left( {i{T_s} - {\tau _u}} \right)}}\\
 \cdot {e^{j2\pi {f_{d,u}}nT}}u\left( {i{T_s} - {\tau _u}} \right)
\end{array} \right]}. \notag
\end{align}

\begin{figure*}[!t]
	\centering
	\includegraphics[width=16cm]{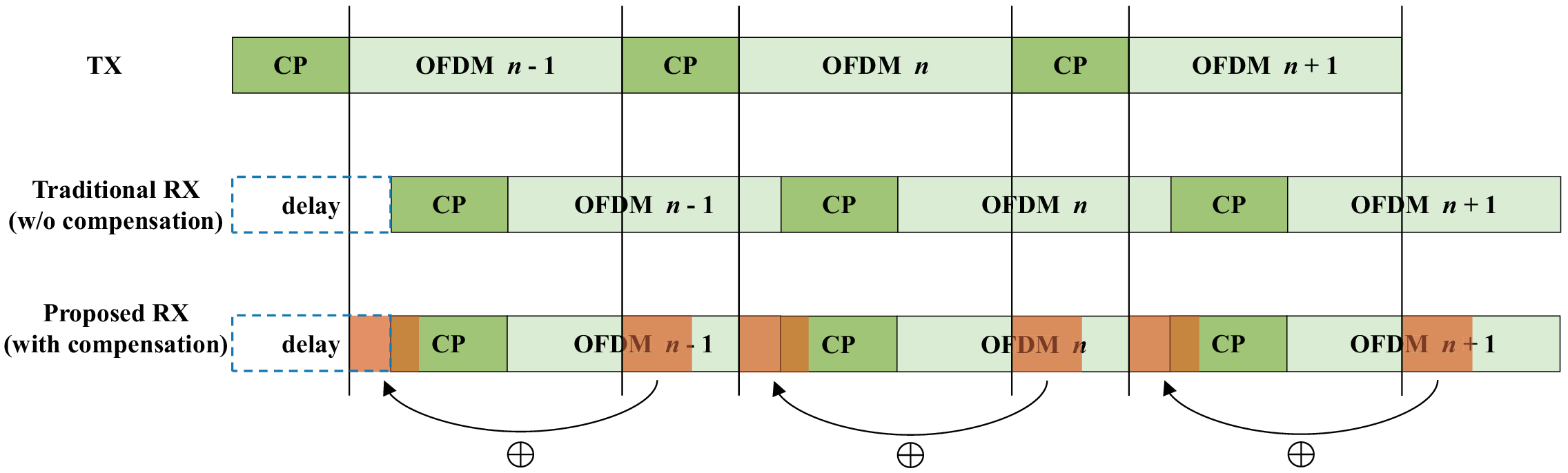}
	\caption{ {The OFDM echo signal model for the long-range target. For the traditional RX, the incomplete OFDM symbol in the receiving window causes ISI and ICI, resulting in reduced SINR. For the proposed RX, the samples (orange parts) after each received OFDM symbol are added to the head of that symbol to achieve coherent compensation.}}
	\label{echo signal model}
	\vspace{-0.3cm}
\end{figure*}

From \eqref{pulse shaping function} and \eqref{sampling of receive signal rewritten}, it is evident that $u\left( {T + i{T_s} - {\tau _u}} \right)$ and $u\left( {i{T_s} - {\tau _u}} \right)$ are nonzero for $i = 0,1,\cdots, N_e  - 1$ and $i = N_e, \cdots ,{N_c} - 1$, respectively, with $N_e={\lfloor {\frac{{{\tau _u} - {T_{cp}}}}{{{T_s}}}} \rceil }$ being the number of samples exceeding the CP duration. Performing FFT on \eqref{sampling of receive signal rewritten}, we obtain the received frequency domain symbol at the $p$th subcarrier of the $n$th OFDM symbol as

\begin{align}
&Y_n^u\left( p \right){\rm{ = }}(1 - \frac{{{N_e}}}{{{N_c}}}){{\tilde \alpha }_u}{{\bf{a}}^T}\left( {{\theta _u}} \right){{\bf{w}}_{tx}}S_{n,p}^{}{e^{ - j2\pi p\Delta f{\tau _u}}}{e^{j2\pi {f_{d,u}}nT}}\notag\\
 &+ \underbrace {\frac{1}{{{N_c}}}{{\tilde \alpha }_u}{{\bf{a}}^T}\left( {{\theta _u}} \right){{\bf{w}}_{tx}}\sum\limits_{i = 0}^{{N_e} - 1} {\sum\limits_{k = 0}^{{N_c} - 1} {\left[ \begin{array}{l}
S_{n - 1,k}^{}{e^{\frac{{j2\pi \left( {k - p} \right)i}}{{{N_c}}}}}\\
 \cdot {e^{\frac{{j2\pi k\left( {{T_{cp}} - {\tau _u}} \right)}}{{{N_c}{T_s}}}}}{e^{j2\pi {f_{d,u}}nT}}
\end{array} \right]} } }_{I_s\left( p \right)}\notag\\
& + \underbrace {\frac{1}{{{N_c}}}{{\tilde \alpha }_u}{{\bf{a}}^T}\left( {{\theta _u}} \right){{\bf{w}}_{tx}}\sum\limits_{i = {N_e}}^{{N_c} - 1} {\sum\limits_{k = 0,k \ne p}^{{N_c} - 1} {\left[ \begin{array}{l}
S_{n,k}^{}{e^{\frac{{j2\pi \left( {k - p} \right)i}}{{{N_c}}}}}\\
 \cdot {e^{\frac{{ - j2\pi k{\tau _u}}}{{{N_c}{T_s}}}}}{e^{j2\pi {f_{d,u}}nT}}
\end{array} \right]} } }_{I_c\left( p \right)}\notag\\
& + {{\bar N}_u}\left( p \right),
\label{frequency domain symbol}
\end{align}
where the first three terms are useful echo signal, ISI, and ICI; the fourth term is the FFT of noise ${\bar n_u}\left( i \right)$. Since ${\bar n_u}\left( i \right)$ is i.i.d, ${\bar N_u}\left( p \right)$ has the same statistical properties as ${\bar n_u}\left( i \right)$. As shown in Fig. \ref{echo signal model}, ISI is due to the echo signal of the previous OFDM symbol falling into the current receiving window. Meanwhile, the OFDM symbol in the receiving window is incomplete. Therefore, the periods of the subcarriers are not an integer multiple in the FFT integration time, which destroys the orthogonality of subcarriers, resulting in ICI. As previously mentioned, we consider that the channel response within an OFDM symbol is time-invariant, so the ICI caused by the Doppler shift is not taken into account.


\subsubsection{Round-trip delay smaller than the CP duration}
In the frequency domain, there is no ISI and ICI when the echo delay is smaller than the CP duration. The received frequency domain symbol is given by
\begin{equation}
	Y_n^u\left( p \right) = {\tilde \alpha _u}{{\bf a}^T}\left( {{\theta _u}} \right){{\bf w}_{tx}}{S_{n,p}}{e^{\frac{{ - j2\pi p{\tau _u}}}{{{N_c}{T_s}}}}}{e^{j2\pi {f_{d,u}}nT}} + {\bar N_u}\left( p \right).
	\label{frequency domain symbol CP doesn't exceed the CP duration}
\end{equation}

\subsection{Coherent compensation-based sensing method}
As shown in \eqref{frequency domain symbol CP doesn't exceed the CP duration}, when the round-trip delay is smaller than the CP duration, the SINR in the frequency domain can be expressed as 
\begin{equation}
	{\Upsilon _{block}^u}{\rm{ = }}\frac{\left| {{{\tilde \alpha }_u}{{\bf{a}}^T}\left( {{\theta _u}} \right){{\bf{w}}_{tx}}} \right|^2\mathbb{E}[ {{{\left| {{S_{n,k}}} \right|}^2}} ]}{{\lambda _u}{\sigma ^2}}={\gamma _0}.
	\label{SINR of short range}
\end{equation}	
 According to \eqref{frequency domain symbol}, when the round-trip delay exceeds the CP duration, the power of the useful signal, ISI, and ICI can be, respectively, expressed as
\begin{equation}
	P_u={\left| {(1 - {\tilde N_e}){{\tilde \alpha }_u}{{\bf{a}}^T}\left( {{\theta _u}} \right){{\bf{w}}_{tx}}} \right|^2}\mathbb{E}[ {{{\left| {{S_{n,p}}} \right|}^2}} ],
	\label{power of useful signal before}
\end{equation}
\vspace{-0.3 cm}
\begin{equation}
	{P_s} = {\tilde N_e}{\left| {{{\tilde \alpha }_u}{{\bf{a}}^T}\left( {{\theta _u}} \right){{\bf{w}}_{tx}}} \right|^2}\mathbb{E}[ {{{\left| {{S_{n - 1,k}}} \right|}^2}} ],
	\label{power of ISI signal before}
\end{equation}
\vspace{-0.3 cm}
\begin{equation}
	{P_c} = {\tilde N_e}(1 - {\tilde N_e}){\left| {{{\tilde \alpha }_u}{{\bf{a}}^T}\left( {{\theta _u}} \right){{\bf{w}}_{tx}}} \right|^2}\mathbb{E}[ {{{\left| {{S_{n,k}}} \right|}^2}} ],
	\label{power of ICI signal before}
\end{equation}
where ${\tilde N_e} = {N_e}/{N_c}$. The proof of \eqref{power of ISI signal before} and \eqref{power of ICI signal before} are provided in {\bf Appendix C}. Therefore, the SINR is given by
\begin{equation}
	{\Upsilon _{block}^u}=\frac{{{{( {1 - {{\tilde N}_e}} )}^2}}}{{{{\tilde N}_e}( {2 - {{\tilde N}_e}} ) + \frac{1}{{{\gamma _0}}}}}.
	\label{SINR of long-range before compensating}
\end{equation}
Comparing \eqref{SINR of short range} with \eqref{SINR of long-range before compensating}, when the round-trip delay exceeds the CP duration, the power ratio of the useful signal to the total echo signal in the frequency domain is reduced to  $({1 - \frac{{{N_e}}}{{{N_c}}}})^2$. Additionally, the introduced ISI and ICI  lead to reduction in the SINR, thereby deteriorating the SINR of RDM and the detection probability.
However, the ISI in the $n$th symbol is coherent with the useful signal in the $(n-1)$th symbol, as shown in \eqref{frequency domain symbol}. Therefore, the power of useful signals can be enhanced through coherent compensation, which improves the SINR of OFDM block in the frequency domain. As shown in Fig. \ref{echo signal model}, the samples after each symbol, which are referred to as compensation samples hereinafter, are added to the head of that symbol in the time domain to achieve coherent compensation. 

Combining \eqref{seperated signal}, the $q$th compensation sample after the $n$th OFDM symbol is given by
\begin{align}
&\bar y_{c,n}^u\left( q \right)=\bar y^u\left( {nT + T_d + q{T_s}} \right)\label{add samples}\\
&{\rm{ = }}\frac{1}{{\sqrt {{N_c}} }}{{\tilde \alpha }_u}{{\bf{a}}^T}\left( {{\theta _u}} \right){{\bf{w}}_{tx}}\sum\limits_{k = 0}^{\mathclap{{N_c} - 1}} {\left[ \begin{array}{l}
S_{n,k}^{}{e^{j2\pi k\Delta f\left( {q{T_s} - {\tau _u}} \right)}}\notag\\
 \cdot {e^{j2\pi {f_{d,u}}\left( {nT + {T_d}} \right)}}\\
 \cdot u\left( {q{T_s} + {T_d} - {\tau _u}} \right)
\end{array} \right]} \\
 &+ \frac{1}{{\sqrt {{N_c}} }}{{\tilde \alpha }_u}{{\bf{a}}^T}\left( {{\theta _u}} \right){{\bf{w}}_{tx}}\sum\limits_{k = 0}^{\mathclap{{N_c} - 1}} {\left[ \begin{array}{l}
{S_{n + 1,k}}{e^{j2\pi {f_d}\left( {nT + {T_d}} \right)}}\\
 \cdot {e^{j2\pi k\Delta f\left( {q{T_s} - {T_{CP}} - {\tau _u}} \right)}}\\
 \cdot u\left( {q{T_s} + {T_d} - T - {\tau _u}} \right)
\end{array} \right]} \notag\\
&{\rm{ + }}{{\bar n}_u}\left( q \right),\notag
\end{align} 
where $q = 0,1, \cdots ,{N_a} - 1$, and $N_a$ is the length of compensation samples.
{{As shown in Fig. \ref{echo signal model} and \eqref{add samples}, the ISI from the $(n+1)$th OFDM symbol will be introduced when the length of compensation samples is greater than the length of offset samples related to the round-trip delay, i.e., $N_a>N_s$, where $N_s=\lfloor {\frac{{{\tau _u}}}{{{T_s}}}} \rceil$. And the SINR of the OFDM block is intuitively lower than that in the case of $N_a=N_s$. Therefore, we mainly focus on the case of  $N_a \le N_s$ in the following analysis to obtain the optimal compensation length \footnote{{Note that the case of $N_a>N_s$ is also considered in the following theoretical analysis and simulations without repeated derivations.}}.}} Adding \eqref{add samples} to \eqref{sampling of receive signal rewritten}, the signal in the time domain after coherent compensation is given by
\begin{align}
			&\tilde y_n^u\left( i \right) = {{\tilde n}_u}\left( i \right) \label{signal in the time domain after coherent compensation}\\
			&+ \frac{1}{{\sqrt {{N_c}} }}{{\tilde \alpha }_u}{{\bf a}^T}\left( {{\theta _u}} \right){{\bf w}_{tx}}\sum\limits_{k = 0}^{\mathclap{{N_c} - 1}} {\left[ \begin{array}{l}
					S_{n - 1,k}^{}{e^{j2\pi k\Delta f\left( {T + i{T_s} - {\tau _u}} \right)}} \notag \\
					\cdot {e^{j2\pi {f_{d,u}}nT}}u\left( {T + i{T_s} - {\tau _u}} \right)
				\end{array} \right]} \\
			&+ \frac{1}{{\sqrt {{N_c}} }}{{\tilde \alpha }_u}{{\bf a}^T}\left( {{\theta _u}} \right){{\bf w}_{tx}}\sum\limits_{k = 0}^{\mathclap{{N_c} - 1}} {\left[ \begin{array}{l}
					S_{n,k}^{}{e^{j2\pi k\Delta f\left( {i{T_s} - {\tau _u}} \right)}}\notag \\
					\cdot {e^{j2\pi {f_{d,u}}\left( {n + 1} \right)T}}u(\frac{{Ti}}{{N_a - 1}} - {T_{cp}})
				\end{array} \right]} \\
			&+ \frac{1}{{\sqrt {{N_c}} }}{{\tilde \alpha }_u}{{\bf a}^T}\left( {{\theta _u}} \right){{\bf w}_{tx}}\sum\limits_{k = 0}^{\mathclap{{N_c} - 1}} {\left[ \begin{array}{l}
					S_{n,k}^{}{e^{j2\pi k\Delta f\left( {i{T_s} - {\tau _u}} \right)}}\notag \\
					\cdot {e^{j2\pi {f_{d,u}}nT}}u\left( {i{T_s} - {\tau _u}} \right)
				\end{array} \right]},
\end{align}
where ${{\tilde n}_u}\left( i \right)$ is the noise after compensation, which can be expressed as
\begin{equation}
	{\tilde n_u}\left( i \right)\sim \left\{ {\begin{array}{*{20}{c}}
			{{\cal C}{\cal N}\left( {0,2{\lambda _u}{\sigma ^2}} \right),}&{0 \le i \le N_a - 1}\\
			{{\cal C}{\cal N}\left( {0,{\lambda _u}{\sigma ^2}} \right),}&{N_a \le i \le {N_c} - 1}
	\end{array}}. \right.
\label{noise in time domain after coherent compensation}
\end{equation}

Performing FFT on \eqref{signal in the time domain after coherent compensation}, the received frequency domain symbol after coherent compensation  can be expressed as
\begin{align}
&\tilde Y_n^u\left( p \right)\label{frequency symbol after coherent compensation}\\
&= ( {1 + {{\tilde N}_a}{e^{j2\pi {f_{d,u}}T}} - {{\tilde N}_e}} )\left[ \begin{array}{l}
{{\tilde \alpha }_u}{{\bf{a}}^T}\left( {{\theta _u}} \right){{\bf{w}}_{tx}}\\
 \cdot S_{n,p}^{}{e^{ - j2\pi p\Delta f{\tau _u}}}{e^{j2\pi {f_{d,u}}nT}}
\end{array} \right]\notag\\
& + \frac{{{\tilde \alpha }_u}{{\bf{a}}^T}\left( {{\theta _u}} \right){{\bf{w}}_{tx}}}{{{N_c}}}\sum\limits_{\scriptstyle k = 0,\hfill\atop
\scriptstyle k \ne p\hfill}^{{N_c} - 1} {\left[ \begin{array}{l}
S_{n,k}^{}{e^{ - j2\pi k\Delta f{\tau _u}}}{e^{j2\pi {f_{d,u}}nT}}\\
 \cdot \left( \begin{array}{l}
{e^{j2\pi {f_{d,u}}T}}\sum\limits_{i = 0}^{N_a - 1} {{e^{\frac{{j2\pi i\left( {k - p} \right)}}{{{N_c}}}}}} \\
+ \sum\limits_{i = {N_e}}^{{N_c} - 1} {{e^{\frac{{j2\pi i\left( {k - p} \right)}}{{{N_c}}}}}} 
\end{array} \right)
\end{array} \right]} \notag\\
 &+ \frac{{{\tilde \alpha }_u}{{\bf{a}}^T}\left( {{\theta _u}} \right){{\bf{w}}_{tx}}}{{{N_c}}}\sum\limits_{i = 0}^{{N_e} - 1} {\sum\limits_{k = 0}^{{N_c} - 1} {\left[ \begin{array}{l}
S_{n - 1,k}^{}{e^{j2\pi k\Delta f\left( {T + i{T_s} - {\tau _u}} \right)}}\\
 \cdot {e^{j2\pi {f_{d,u}}nT}}{e^{\frac{{ - j2\pi ip}}{{{N_c}}}}}
\end{array} \right]} } \notag\\
& + \tilde N_n^u\left( p \right),\notag
\end{align}
where ${\tilde N_a}=N_a/N_c$. The power of the useful signal, ICI, and noise after coherent compensation can be, respectively, expressed as 
\begin{align}
	{\tilde P_u} = &\left[ {{{( {{1} - {\tilde N_e}} )}^2} + {{\tilde N_a}^2} + 2{\tilde N_a}( {{1} - {\tilde N_e}} )\cos \left( {2\pi {f_{d,u}}T} \right)}\right]\notag \\
	&\cdot {\left| {{{\tilde \alpha }_u}{{\bf a}^T}\left( {{\theta _u}} \right){{\bf w}_{tx}}} \right|^2}\mathbb{E}\left[ {{{\left| {S_{n,p}^{}} \right|}^2}} \right],
 \label{power of useful signal after}
\end{align}
\vspace{-0.3cm}
\begin{align}
	{{\tilde P}_c} =& \left[ \begin{array}{l}
		{\tilde N_a}( {1 - {\tilde N_a}} ) + {\tilde N_e}( {1 - {\tilde N_e}})\\
		+ 2( {{\tilde N_a}{\tilde N_e} - \min [ {{\tilde N_a},{\tilde N_e}} ]} )\cos \left( {2\pi {f_{d,u}}T} \right)
	\end{array} \right]\notag \\
	&\cdot{\left| {{{\tilde \alpha }_u}{{\bf a}^T}\left( {{\theta _u}} \right){{\bf w}_{tx}}} \right|^2}\mathbb{E}\left[ {{{\left| {S_{n,k}^{}} \right|}^2}} \right],
\end{align}
\vspace{-0.3cm}
\begin{equation}
	{\tilde P_n} = ( {1 + {{\tilde N}_a}}){\lambda _u}{\sigma ^2}.
	\label{power of noise after coherent compensation}
\end{equation}

The proof of \eqref{power of noise after coherent compensation} is provided in {\bf Appendix D}. Comparing \eqref{power of useful signal before} with \eqref{power of useful signal after}, it is evident that the power of the useful signal improves after coherent compensation. Furthermore, the power of ICI is minimum when ${\tilde N_a}={\tilde N_e}$.  Compared with the noise in \eqref{frequency domain symbol}, the power of noise in \eqref{frequency symbol after coherent compensation} increases due to the noise samples introduced by compensation, as shown in \eqref{power of noise after coherent compensation}.
Combining \eqref{power of ISI signal before} and \eqref{power of useful signal after}-\eqref{power of noise after coherent compensation}, the SINR after coherent compensation can be expressed as \eqref{SINR after coherent compensation} at the top of next page.

{For the case of $N_a > N_s$, the SINR after coherent compensation can be expressed as \eqref{SINR after coherent compensation for case of Na > Ns} at the top of next page, where ${\tilde N}_s=N_s/N_c$. Since the derivation is similar to the case of  $N_a \le N_s$, we omit it for brevity.}
\section{Performance Analysis}
\begin{figure*}
	\begin{equation}
		{\Upsilon _{c, block}^u} = \frac{{\left[ {{{( {1 - {{\tilde N}_e}} )}^2} + \tilde N_a^2 + 2{{\tilde N}_a}( {1 - {{\tilde N}_e}} )\cos \left( {2\pi {f_{d,u}}T} \right)} \right]}}{{{{\tilde N}_a}\left[ {1 - {{\tilde N}_a} + 2{{\tilde N}_e}\cos \left( {2\pi {f_{d,u}}T} \right)} \right] + {{\tilde N}_e}( {2 - {{\tilde N}_e}} ) - 2\min \left[ {{{\tilde N}_a},{{\tilde N}_e}} \right]\cos \left( {2\pi {f_{d,u}}T} \right) + ( {1 + {{\tilde N}_a}})\frac{1}{{{\gamma _0}}}}}.
		\label{SINR after coherent compensation}
	\end{equation}
	
	{\noindent} \rule[-10pt]{18cm}{0.05em}
\end{figure*}
\label{sec: Performance Analysis}
In this section, the performance of the proposed method is analyzed. Firstly, the SINR of the OFDM block is derived and the optimal length of compensation samples is analyzed. Subsequently, the theoretical SINR of RDM is derived, which is related to sensing performance metrics such as detection probability. Finally, the maximum sensing range supported by the proposed method is clarified.
\subsection{SINR of OFDM block}
As shown in \eqref{SINR after coherent compensation}, when the length of compensation samples is less than $N_e$, i.e., $0<N_a<{N}_e$ (or $0<{\tilde N}_a<{\tilde N}_e$), the SINR of OFDM block can be expressed as  
\begin{align}
    {\gamma _2} = \frac{{{{(1 - {{\tilde N}_e})}^2} + \tilde N_a^2 + 2{{\tilde N}_a}(1 - {{\tilde N}_e})\cos (2\pi {f_{d,u}}T)}}{{h_1( {{{\tilde N}_a}})}},
    \label{SINR after coherent compensation length 1}
\end{align}
where $h_1( {{{\tilde N}_a}} )$ is the power of interference plus noise after coherent compensation, which is given by
\begin{align}
h_1( {{{\tilde N}_a}} ) &= {{\tilde N}_a}\left[ {(1 - {{\tilde N}_a}) - 2(1 - {{\tilde N}_e})\cos (2\pi {f_{d,u}}T)} \right]\\
 &+ {{\tilde N}_e}(2 - {{\tilde N}_e}) + (1 + {{\tilde N}_a})\frac{1}{{{\gamma _0}}}. \notag
\end{align}
For the convenience of analysis, we treat ${\tilde N}_a$ as a continuous variable. Therefore, the first-order derivative of \eqref{SINR after coherent compensation length 1} with respect to ${\tilde N}_a$ is given by 
\begin{align}
    \frac{\partial }{{\partial {N_a}}}{\gamma _2} = \frac{{(\frac{1}{{{\gamma _0}}} + 1)f( {{{\tilde N}_a}} )}}{{{{\left[ {h_1( {{{\tilde N}_a}} )} \right]}^{\rm{2}}}}},\label{derivatinve of SINR length 1}
\end{align}
where
\begin{equation}
	f( {{{\tilde N}_a}} ) = \tilde N_a^2 + 2{\tilde N_a} - {( {1 - {{\tilde N}_e}} )^2} + 2( {1 - {{\tilde N}_e}} )\cos ( {2\pi {f_{d,u}}T} ).
\end{equation}
Moreover, $f( {{{\tilde N}_a}} )$ increases monotonically within the interval $\left[ {0,{N_e}} \right]$. Thus, we obtain
\begin{align}
		f( {{{\tilde N}_a}} ) \ge  - {( {1 - {{\tilde N}_e}} )^2} + 2( {1 - {{\tilde N}_e}} )\cos ( {2\pi {f_{d,u}}T} )\notag \\
		{\rm{ = }}( {1 - {{\tilde N}_e}} )\left[ {2\cos \left( {2\pi {f_{d,u}}T} \right) - 1 + {{\tilde N}_e}} \right].
\end{align}
Since the Doppler frequency shift is much smaller than the subcarrier spacing, $\cos \left( {2\pi {f_{d,u}}T} \right) \approx 1$ is obtained. Besides, we can derive ${{{\tilde N}_e}}<1$ from $N_e={\lfloor {\frac{{{\tau _u} - {T_{cp}}}}{{{T_s}}}} \rceil }$. Therefore, $f( {{{\tilde N}_a}} ) > 0$ holds, and the SINR of OFDM block increases monotonically when $0<N_a<{N}_e$.

When the length of compensation samples is greater than $N_e$, i.e., $N_e \le N_a \le N_s$ (or ${\tilde N}_e \le {\tilde N}_a \le {\tilde N}_s$), the SINR of OFDM block can be expressed as 
\begin{align}
    {\gamma _2} = \frac{{{{(1 - {{\tilde N}_e})}^2} + \tilde N_a^2 + 2{{\tilde N}_a}(1 - {{\tilde N}_e})\cos (2\pi {f_{d,u}}T)}}{{{h_2}({{\tilde N}_a})}},
    \label{SINR after coherent compensation length 2}
\end{align}
where ${h_2}({{\tilde N}_a})$ is the power of interference plus noise after coherent compensation, which is given by
\begin{align}
    {h_2}({{\tilde N}_a}) &= (1 - {{\tilde N}_a})\left[ {{{\tilde N}_a} - 2{{\tilde N}_e}\cos \left( {2\pi {f_{d,u}}T} \right)} \right] \\
    &+ {{\tilde N}_e}(2 - {{\tilde N}_e}) + (1 + {{\tilde N}_a})\frac{1}{{{\gamma _0}}}.\notag
\end{align}
The first-order derivative of \eqref{SINR after coherent compensation length 2} with respect to ${{{\tilde N}_a}}$ is further obtained as
\begin{equation}
    \frac{\partial }{{\partial {{\tilde N}_a}}}{\gamma _2} = \frac{{g({{\tilde N}_a})}}{{{{\left[ {{h_2}({{\tilde N}_a})} \right]}^2}}},
    \label{derivatinve of SINR length 2}
\end{equation}
where $g({{\tilde N}_a})$ is given by
\begin{align}
&g({{\tilde N}_a}) = (\frac{1}{{{\gamma _0}}} + 1 + 2B)\tilde N_a^2 + (\frac{2}{{{\gamma _0}}} + 2 - 4{{\tilde N}_e}B){{\tilde N}_a}\\
& + (1 - {{\tilde N}_e})\left[ {\frac{1}{{{\gamma _0}}}(2B - 1) + ( {\frac{1}{{{\gamma _0}}} + 2B - 4{B^2} + 1} ){{\tilde N}_e} - 1} \right],
\notag
\end{align}
where $B=\cos\left( {2\pi {f_{d,u}}T} \right)$. Since ${g( {{{\tilde N}_a}} )}$ is a quadratic function of ${{\tilde N}_a}$, the  axis of symmetry can be expressed as
\begin{equation}
	{\tilde N_a} =  - \frac{{\frac{1}{{{\gamma _0}}} + 1 - 2{{\tilde N}_e}B}}{{\frac{1}{{{\gamma _0}}} + 1 + 2B}}<{\tilde N_e}.
\end{equation}
According to the properties of quadratic functions, the optimal length of compensation samples to maximize the SINR of OFDM block is $N_a=N_e$ or $N_a=N_s$. 

{{For the case of $N_s< N_a \le {N}_c$ (or ${\tilde N}_s < {\tilde N}_a \le 1$), the SINR of the OFDM block decreases with the increase of $N_a$, as shown in \eqref{SINR after coherent compensation for case of Na > Ns} at the top of next page.  In summary, we have the following conclusions.}}
\begin{figure*}
	{\begin{equation}
			{\tilde \Upsilon _{c,block}^u = \frac{{\left[ {{{(1 - {{\tilde N}_e})}^2} + {{\tilde N}_s}^2 + 2{{\tilde N}_s}(1 - {{\tilde N}_e})\cos \left( {2\pi {f_{d,u}}T} \right)} \right]}}{{\left[ {{{\tilde N}_s}(1 - {{\tilde N}_s}) + {{\tilde N}_e}(2 - {{\tilde N}_e}) + ( {{{\tilde N}_a} - {{\tilde N}_s}} ) + 2({{\tilde N}_s}{{\tilde N}_e} - {{\tilde N}_e})\cos \left( {2\pi {f_{d,u}}T} \right)} \right] + (1 + {{\tilde N}_a})\frac{1}{{{\gamma _0}}}}}.}
			\label{SINR after coherent compensation for case of Na > Ns}
	\end{equation}}
\end{figure*}
\begin{remark}
When the length of compensation samples is less than $N_e$, i.e., $0<N_a<{N}_e$, the SINR of the OFDM block is monotonically increasing with  $N_a$ increasing. 
\label{remark 1}
\end{remark}
\begin{remark}
The optimal length of compensation samples is $N_a=N_e$ or $N_a=N_s$. When $N_a=N_e$, a complete OFDM symbol is formed in the receiving time interval, ICI is minimized. Furthermore, the power of useful signal in the frequency domain is maximized when $N_a=N_s$.
\label{remark 2}
\end{remark}

Due to its low computational complexity, the 2D-FFT method is commonly utilized to estimate the range and velocity of targets \cite{Sturm 1}. The 2D-FFT method performs IFFT and FFT along the subcarrier and symbol dimensions, respectively, to obtain RDM for target detection and parameter estimation. Next, the SINR of RDM obtained by the traditional sensing method and the proposed sensing method is analyzed.
\begin{figure*}
	\begin{align}
			RD{M^u}\left( {k,l} \right) &= \frac{1}{{M{N_c}}}\sum\limits_{n = 0}^{M - 1} {\sum\limits_{p = 0}^{{N_c} - 1} {\frac{1}{{S_{n,p}^{}}}\hat Y_n^u\left( p \right){e^{j\frac{{2\pi pk}}{{{N_c}}}}}{e^{-j\frac{{2\pi nl}}{M}}}} }\notag \\
			&= \frac{1}{{M{N_c}}}\sum\limits_{n = 0}^{M - 1} {\sum\limits_{p = 0}^{{N_c} - 1} {\left[ \begin{array}{l}
						({1+{{\tilde N}_a}{e^{j2\pi {f_{d,u}}T}}  - {{\tilde N}_e}}){{\tilde \alpha }_u}{{\bf a}^T}\left( {{\theta _u}} \right){{\bf w}_{tx}}{e^{ - j2\pi p\Delta f{\tau _u}}}{e^{j2\pi {f_{d,u}}nT}} \\
						+ \frac{1}{{S_{n,p}^{}}} I_{c,n}^u\left( p \right) + \frac{1}{{S_{n,p}^{}}}I_{s,n}^u\left( p \right) + \frac{1}{{S_{n,p}^{}}}{{\tilde N}_u}\left( p \right)
					\end{array} \right]{e^{j\frac{{2\pi pk}}{{{N_c}}}}}{e^{-j\frac{{2\pi nl}}{M}}}} } .
				\label{RDM coherent compensation}\tag{51}
	\end{align}
\end{figure*}
\begin{figure*}
	\begin{align}
		\mathbb{V}&\left[RD{M^u}\left( {k,l} \right)\right]  = \frac{1}{{M{N_c}}}\mathbb{E}\left[ {{{\left| {\frac{1}{{S_{n,k}^{}}}} \right|}^2}} \right]\left[ {( {A + {{\tilde N}_e}} ){{\left| {{{\tilde \alpha }_u}{{\bf a}^T}\left( {{\theta _u}} \right){{\bf w}_{tx}}} \right|}^2}\mathbb{E}\left[ {{{\left| {S_{n,k}^{}} \right|}^2}} \right] + ( {1 + {{\tilde N}_a}} ){\lambda _u}{\sigma ^2}} \right]\notag \\
		& \mathop  = \limits^{(a)} \frac{1}{{M{N_c}}}\mathbb{E}\left[ {{{\left| {\frac{1}{{S_{n,k}^{}}}} \right|}^2}} \right]\left\{ {\left[ {{{\tilde N}_e} + | {{{\tilde N}_a} - {{\tilde N}_e}} |( {1 - | {{{\tilde N}_a} - {{\tilde N}_e}} |} )} \right]{{| {{{\tilde \alpha }_u}{{\bf a}^T}\left( {{\theta _u}} \right){{\bf w}_{tx}}} |}^2}\mathbb{E}\left[ {{{\left| {S_{n,k}^{}} \right|}^2}} \right] + ( {1 + {{\tilde N}_a}} ){\lambda _u}{\sigma ^2}} \right\}.
		\label{power of IN}\tag{52}
	\end{align}
\end{figure*}
\begin{figure*}
	{\begin{equation}
		{\tilde \Upsilon} _{c,RDM}^u = 1 + \frac{{{{\left| {(1 + {{\tilde N}_s} - {{\tilde N}_e}){{\tilde \alpha }_u}{{\bf{a}}^T}\left( {{\theta _u}} \right){{\bf{w}}_{tx}}} \right|}^2}}}{{\frac{1}{{M{N_c}}}\mathbb{E}\left[ {{{\left| {\frac{1}{{S_{n,k}^{}}}} \right|}^2}} \right]\left\{ {\left[ {{{\tilde N}_a} - {{({{\tilde N}_s} - {{\tilde N}_e})}^2}} \right]|{{\tilde \alpha }_u}{{\bf{a}}^T}\left( {{\theta _u}} \right){{\bf{w}}_{tx}}{|^2}\mathbb{E}\left[ {{{\left| {S_{n,k}^{}} \right|}^2}} \right] + (1 + {{\tilde N}_a}){\lambda _u}{\sigma ^2}} \right\}}}.
		\label{SINR of RDM for case of Na > Ns}\tag{55}
	\end{equation}}
{\noindent} \rule[-10pt]{18cm}{0.05em}
\end{figure*}
\begin{figure*}
	\begin{equation}
		\Upsilon_{t,RDM}^u \approx 1 + \frac{{M{N_c}{{\left( {1 - {\tilde N}_e^u} \right)}^2}{{\left| {{{\tilde \alpha }_p}{\bf w}_{rx}^T{\bf b}\left( {{\theta _u}} \right){{\bf a}^T}\left( {{\theta _u}} \right){{\bf w}_{tx}}} \right|}^2}}}{{\sum\limits_{p = 0}^{U - 1} {{\tilde N}_e^p\left( {2 - {\tilde N}_e^p} \right){{\left| {{{\tilde \alpha }_p}{\bf w}_{rx}^T{\bf b}\left( {{\theta _p}} \right){{\bf a}^T}\left( {{\theta _p}} \right){{\bf w}_{tx}}} \right|}^2}\mathbb{E}\left[ {{{\left| {\frac{1}{{{S_{n - 1,k}}}}} \right|}^2}} \right]}  + {\bf w}_{rx}^T{\bf w}_{rx}^*{\sigma ^2}\mathbb{E}\left[ {{{\left| {\frac{1}{{{S_{n - 1,k}}}}} \right|}^2}} \right]}}.
		\label{SINR of RDM for traditional sensing algorithm}\tag{57}
	\end{equation}
	{\noindent} \rule[-10pt]{18cm}{0.05em}
\end{figure*}
\subsection{SINR of RDM}
In this section, the SINRs of RDM obtained by the proposed coherent compensation-based sensing method, signal separation method, and traditional 2D-FFT method \cite{Sturm 1} are derived. Since the ISI and ICI are no longer i.i.d, and coherent compensation introduces correlation between adjacent OFDM blocks, it is challenging to obtain a closed-form solution for SINR of RDM. To tackle this problem, an upper bound for the SINR of RDM is derived.

 For the case of $N_a \le N_s$, the frequency symbol after coherent compensation in \eqref{frequency symbol after coherent compensation} can be rewritten as
\begin{align}
	&\hat Y_n^u\left( p \right)\notag \\
	&{\rm{ = }}({1+{{\tilde N}_a}{e^{j2\pi {f_{d,u}}T}}  - {{\tilde N}_e}}){\tilde \alpha _u}{{\bf a}^T}({\theta _u}){{\bf w}_{tx}}S_{n,p}{e^{ - j2\pi p\Delta f{\tau _u}}}\notag \\
	&\qquad \qquad \qquad \qquad \qquad \qquad \qquad \qquad \qquad \times{e^{j2\pi {f_{d,u}}nT}} \notag \\
	&+ I_{c,n}^u\left( p \right) + I_{s,n}^u\left( p \right) + \tilde N_n^u\left( p \right). 
\end{align}
According to the conclusions in \cite{wu VCP}\cite{Convergence of Complex Envelope}, $I_{s,n}^u(p)$ and $I_{c,n}^u(p)$  follow the CSCG distributions. Thus, we obtain
\begin{equation}
	I_{s,n}^u\left( p \right)\sim {\cal C}{\cal N}\left( {0,{{\tilde N}_e}{{\left| {{{\tilde \alpha }_u}{{\bf a}^T}\left( {{\theta _u}} \right){{\bf w}_{tx}}} \right|}^2}E\left[ {{{\left| {{S_{n - 1,k}}} \right|}^2}} \right]} \right),
\end{equation}
\begin{equation}
	I_{c,n}^u\left( p \right)\sim {\cal C}{\cal N}\left( {0,A{{\left| {{{\tilde \alpha }_u}{{\bf a}^T}\left( {{\theta _u}} \right){{\bf w}_{tx}}} \right|}^2}E\left[ {{{\left| {S_{n,k}^{}} \right|}^2}} \right]} \right),
\end{equation}
where 
\begin{align}
	A&={{\tilde N}_a}( {1 - {{\tilde N}_a}} ) + [ {2{{\tilde N}_a}{{\tilde N}_e} - 2\min ( {{{\tilde N}_e},{{\tilde N}_a}} )} ]\cos \left( {2\pi {f_{d,u}}T} \right) \notag \\
	&+ {{\tilde N}_e}( {1 - {{\tilde N}_e}} ).
	\label{A}
\end{align}
Then, the value of the $(k,l)$th bin in RDM obtained by the 2D-FFT method can be expressed as \eqref{RDM coherent compensation} at the top of this page.

The SINR of RDM is defined as the power ratio of the bin corresponding to the target to the interference plus noise \cite{yuan multiple CP}. We ignore the correlation of $I_{c,n}^u(p)$ and $I_{s,n}^u(p)$ along $n$ and $p$ and derive an approximate theoretical solution for the SINR of RDM, which is an upper bound. The variance of $RD{M^u}\left( {k,l} \right)$ is the power of interference plus noise in RDM, as shown in \eqref{power of IN} at the top of this page, where $(a)$ is obtained since the Doppler frequency shift is much smaller than the subcarrier spacing, i.e., $\cos \left( {2\pi {f_{d,u}}T} \right) \approx 1$.

For the $u$th target, the power of the useful signal in $RDM^u(k,l)$ is the power of $(k_u,l_u)$th bin,  which is
\begin{align}
	&\mathbb{E}\left[\left|RD{M^u}\left( {k_u,l_u} \right)\right|^2\right] \tag{53} \\
	&\approx {\left| {( {1 + {{\tilde N}_a} - {{\tilde N}_e}} ){{\tilde \alpha }_u}{{\bf a}^T}\left( {{\theta _u}} \right){{\bf w}_{tx}}} \right|^2}+\mathbb{V}\left[RD{M^u}\left( {k_u,l_u} \right)\right], \notag
\end{align}
where $k_u=\lfloor B\tau_u \rceil$ and $l_u=\lfloor Mf_{d,u}T\rceil$. Therefore, the SINR of RDM after coherent compensation for the case of $N_a \le N_s$ is
\begin{align}
		{\Upsilon_{c,RDM}^u} &= \frac{{\mathbb{E}\left[ {{{\left| {RD{M^u}\left( {{k_u},{l_u}} \right)} \right|}^2}} \right]}}{{\mathbb{V}\left[ {RD{M^u}\left( {{k_u},{l_u}} \right)} \right]}} \label{SINR of RDM for coherent compensation algorithm}\tag{54}\\
		&= 1 + \frac{{{{\left| {(1 + {{\tilde N}_a} - {{\tilde N}_e}){{\tilde \alpha }_u}{{\bf{a}}^T}\left( {{\theta _u}} \right){{\bf{w}}_{tx}}} \right|}^2}}}{{\mathbb{V} \left[ {RD{M^u}\left( {k_u,l_u} \right)} \right]}}. \notag
\end{align}
{Similarly, the SINR of RDM after coherent compensation for the case of $N_s < N_a \le N_c$ can be expressed as \eqref{SINR of RDM for case of Na > Ns} at the top of this page.}

For the signal separation method, the SINR of RDM can be obtained by substituting ${\tilde N}_a=0$ into \eqref{SINR of RDM for coherent compensation algorithm}, which can be expressed as
\begin{align}
		&\Upsilon_{s,RDM}^u  \label{SINR of RDM for signal separation algorithm}\tag{56}\\
		&\approx 1 + \frac{{M{N_c}{{( {1 - {{\tilde N}_e}} )}^2}}}{{{{\tilde N}_e}( {2 - {{\tilde N}_e}} )\mathbb{E}\left[ {{{\left| {S_{n,k}^{}} \right|}^2}} \right]\mathbb{E}\left[ {{{\left| {\frac{1}{{S_{n,k}^{}}}} \right|}^2}} \right] + \frac{{{\lambda _u}{\sigma ^2}\mathbb{E}\left[ {{{\left| {\frac{1}{{S_{n,k}^{}}}} \right|}^2}} \right]}}{{{{\left| {{{\tilde \alpha }_u}{{\bf a}^T}\left( {{\theta _u}} \right){{\bf w}_{tx}}} \right|}^2}}}}}.\notag
\end{align}

If the traditional received signal expressed in \eqref{traditional received signal} is used for sensing, the interference from multiple targets are superimposed in the RDM. Thus, the SINR of RDM for the traditional 2D-FFT method is given by \eqref{SINR of RDM for traditional sensing algorithm} at the top of next page.

\subsection{The maximum sensing range}
The 2D-FFT method is finally carried out for range and velocity estimation, whose maximum unambiguous range is given by \cite{Sturm 1}
\begin{equation}
	{R_{unambi}} = \frac{{{c_0}}}{{2\Delta f}} = \frac{{{c_0}{T_d}}}{2}.\tag{58} 
\end{equation}
However, the 2D-FFT method assumes that the round-trip delay does not exceed the CP duration \cite{Sturm 1}, so the actual maximum sensing range limited by CP length is much smaller than the maximum unambiguous range. This paper proposes a coherent compensation method to improve the SINR of the OFDM block, extending the maximum sensing range to approximate the maximum unambiguous range of the 2D-FFT method.

\section{Simulation Results and Analysis}
\label{sec: Simulation Results and Analysis}
This section conducts extensive simulations to verify the aforementioned analysis on SINR of RDM and the detection probability. The parameter setup is shown in TABLE \ref{sys_para} unless otherwise emphasized. For convenience, the abbreviations of the proposed sensing method and other benchmarking methods are provided as follows.
\begin{enumerate}
        \item[$\bullet$] The \emph{2D-FFT} method first utilizes LS beamforming to combine received signals from different antennas, as shown in \eqref{traditional received signal}. Then, the point-division-based 2D-FFT method  \cite{Sturm 1} and constant false alarm rate (CFAR) detector \cite{CFAR} are applied for target detection and parameter estimation.
	\item[$\bullet$] The \emph{signal separation and coherent compensation  (S \& C)} method first utilizes the MUSIC method to estimate the angles of targets within one beam, followed by LS estimation for signal separation. The coherent compensation is then performed to enhance the SINR of the OFDM block. Finally, the point-division-based 2D-FFT method and CFAR detector are applied for target detection and parameter estimation.
	\item[$\bullet$] The \emph{benchmark case} solely achieves the LS estimation-based signal separation without coherent compensation. The other steps remain the same as the S \& C method.
\end{enumerate}

\begin{table}[h]
	\caption{Simulation parameters}
        \label{sys_para}
	\begin{center}
		\begin{tabular}{l l l}
			\hline
			\hline
			
			{Parameter} & {Symbol} & {Value} \\
			
			\hline
			
			Carrier frequency & ${f_c}$ & $28$ GHz\\
			
			Subcarrier spacing & $\Delta f$ & $120$ kHz\\
			
			Number of subcarriers & $N_c$ & $4096$\\
			
			Number of OFDM symbols & $M$ & $256$\\
			
			Elementary OFDM symbol duration & $T_{d}$ & $8.33$ $\mathrm{\mu s}$\\
			
			CP duration & $T_{cp}$ & $0.59$ $\mathrm{\mu s}$\\
			
			Total OFDM symbol duration & $T$ & $8.92$ $\mathrm{\mu s}$\\
			
			Number of transmit antennas &
			$N_t$ & $16$\\
			
			Number of receive antennas &
			$N_r$ & $16$\\
			
			Gain of transmit antennas &
			$G_t$ & $32$ dB\\
			
			Gain of receive antennas &
			$G_r$ & $32$ dB\\
			
			RCS & $\alpha$ & $10$ ${\text m}^2$\\
			
			Height of BS & $h$ & $30$ m\\
			
			Light speed & $c_0$ & $3 \times 10^8$ m/s\\
			\hline
			\hline
		\end{tabular}
	\end{center}
	\label{table_1}
\end{table}

In the following simulations, the transmit power of BS is set as $46$ dBm, and the $16$ QAM symbols are transmitted. The transmit beam is aligned with the UE at $500$ m, and the round-trip delay is greater than the CP duration. Due to the wide beam, BS will also cover the close-range target, which may drown the long-range target, especially when the ISI and ICI also exist in the echo signals of the close-range target. Therefore, we simultaneously consider the close-range and long-range targets in the simulations.  The beamwidth is about $6.36^\circ$ shown in Fig. \ref{figure_2_c}, which results in a transmit beam that can cover the target at $260$ m. The velocities of the targets are set as $40$ m/s and $60$ m/s, respectively.

Fig. \ref{DoA estimation music} shows the DoA estimation results using the MUSIC method. The MUSIC method demonstrates the capability to distinguish targets within one beam with considerable precision. The reduction in the search range of angle is due to the known prior information about the transmit beam, leading to a substantial reduction in computational complexity.
\begin{figure}[ht]
	\centering
	\includegraphics[scale=0.64]{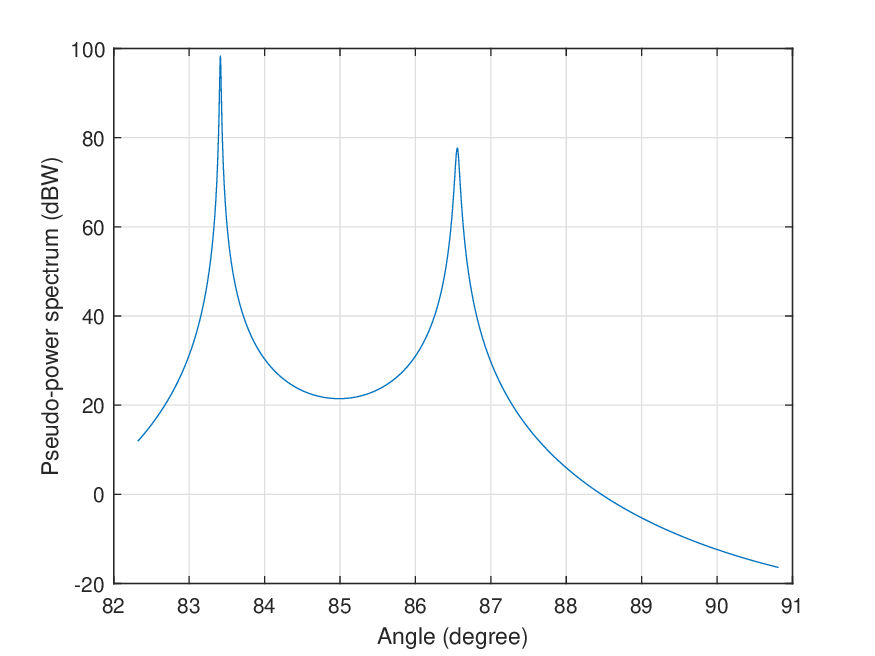}
	\caption{{The DoA estimation using the MUSIC method.}}
	\label{DoA estimation music}
\end{figure}

\begin{figure}[ht]
	\centering
	\includegraphics[scale=0.64]{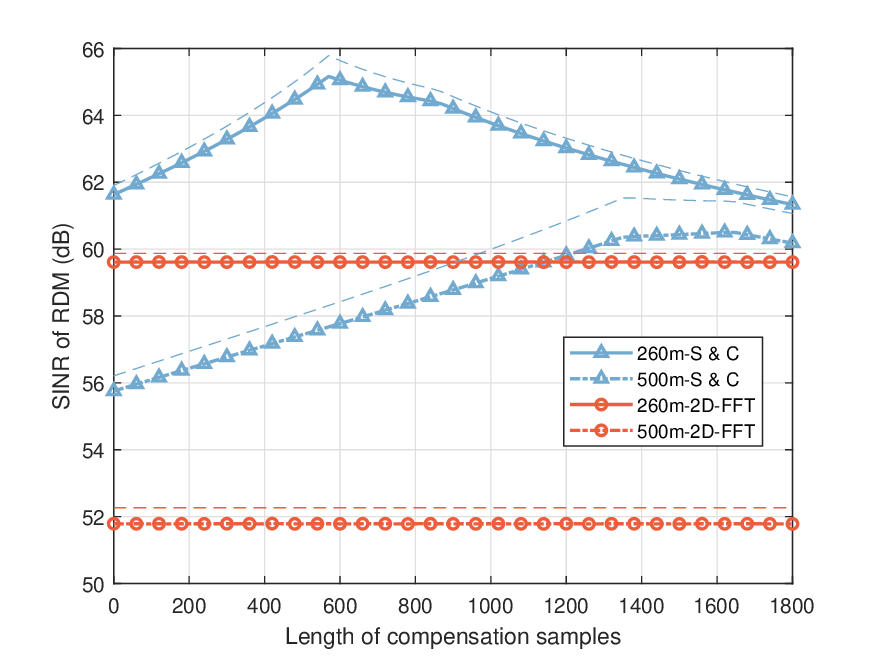}
	\caption{{The SINR of RDM versus the length of compensation samples.}}
	\label{The SINR of RDM versus the length of compensation samples}
\end{figure}

Fig. \ref{The SINR of RDM versus the length of compensation samples} depicts the SINR of RDM with different lengths of compensation samples. {The dash curves are the theoretical results, which fit the simulation results well, verifying the validity of the theoretical derivation.} When the length of compensation samples is less than the number of offset samples corresponding to the round-trip delay, the SINR of RDM monotonically increases with the increasing of the length of compensation samples, as anticipated in {\bf Remark \ref{remark 1}}. Moreover, there exists an optimal compensation length maximizing the SINR of RDM, as analyzed in {\bf Remark \ref{remark 2}}. For the target at $260$ m, the optimal compensation length is $N_e$, while it is $N_s$ for the target at $500$ m. This optimality criterion is confirmed in Fig. \ref{The SINR of RDM versus the length of compensation samples}.

In  Fig. \ref{The SINR of RDM versus the transmit power}, the SINR of RDM is presented for different transmit power. 
The simulated results of the S \& C, benchmark case, and 2D-FFT methods are obtained via $1000$ Monte Carlo simulations. The theoretical results obtained in Section IV-B are plotted as dash curves, which fit the simulation results well. As the transmit power increases, the SINR of RDM converges at high transmit power. Since the power of the useful signal, ISI, and ICI all increase with the increasing of transmit power, which ultimately limits the SINR of RDM from increasing infinitely. This phenomenon is consistent with the theoretical results in \eqref{SINR of RDM for coherent compensation algorithm}, \eqref{SINR of RDM for signal separation algorithm}, and \eqref{SINR of RDM for traditional sensing algorithm}.

In Fig. \ref{The SINR of RDM versus the transmit power}, the SINR of RDM obtained by the traditional 2D-FFT method does not change obviously as the transmit power increases. This is because the power of ISI and ICI also amplifies with the increase of transmit power. 
Moreover, the ISI and ICI of targets at different ranges are superimposed under the traditional 2D-FFT method, as shown in \eqref{SINR of RDM for traditional sensing algorithm}. The benchmark case effectively separates the echo signals reflected from different targets, as well as the ISI and ICI. Consequently, the SINR of the RDM achieved by the benchmark case shows a notable enhancement compared with the traditional 2D-FFT method. However, the signal separation method based on LS estimation can also amplify the noise. In the case where the echo signal is weak, the reduction of the interference signal from other targets may not compensate for the noise amplification, thus deteriorating the SINR of RDM, as shown in Fig. \ref{The SINR of RDM versus the transmit power}. Furthermore, the SINR improvement achieved by the benchmark case is more prominent for long-range targets. This is because the power of ISI and ICI for close-range targets is larger, which has a severe impact on long-range target. Based on the spatial signal separation method, the power of the useful signal (i.e., the peak power in RDM) is further enhanced through the coherent compensation, thus increasing the SINR of RDM. As shown in Fig. \ref{The SINR of RDM versus the transmit power}, the SINR of RDM  obtained by the S \& C method increases approximately $10$ dB for the target at $500$ m and about $6$ dB for the target at $260$ m compared with the 2D-FFT method. This confirms the effectiveness of the proposed long-range sensing method. 
\begin{figure}[t]
	\centering
	\includegraphics[scale=0.64]{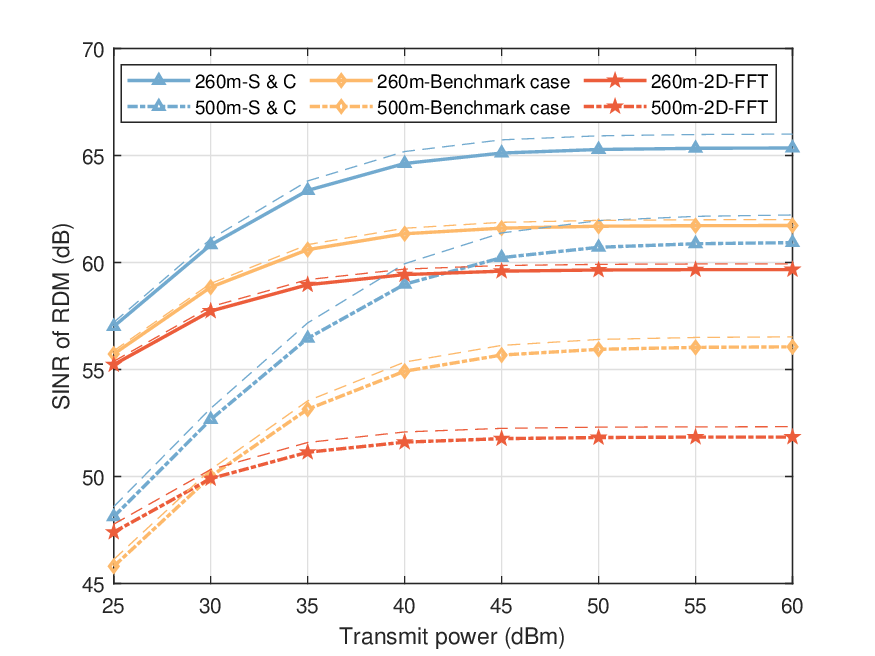}
	\caption{{The SINR of RDM versus the transmit power. The dash curves are theoretical results, while the other curves are simulation results.}}
	\label{The SINR of RDM versus the transmit power}
\end{figure}

Fig. \ref{The SINR of RDM versus the different range of target} demonstrates the SINR of RDM for S \& C, benchmark case, and traditional 2D-FFT methods under various ranges. The target at each range suffers interference from the closest target within the same beam. For these three methods, the SINR of RDM decreases as the range of the target increases. Moreover, both the  S \& C and benchmark case methods notably enhance the SINR of RDM for the target at a larger range, which reveals the benefit of the proposed method to improve the sensing performance for long-range targets. For close-range targets, which are less susceptible to ISI and ICI, the proposed method yields a modest improvement in the SINR of RDM. Besides, the derived theoretical values for the SINR of RDM fit better with the simulation values particularly when the range is longer.

\begin{figure}[t]
	\centering
	\includegraphics[scale=0.64]{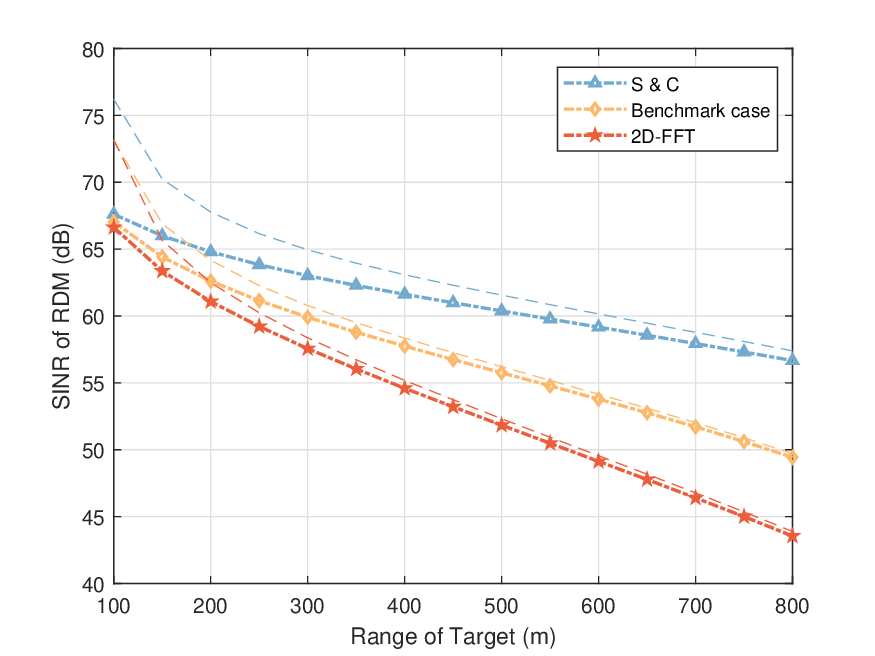}
	\caption{{The SINR of RDM versus the different range of the target.}}
	\label{The SINR of RDM versus the different range of target}
\end{figure}

\begin{figure}[t]
	\centering
	\includegraphics[scale=0.64]{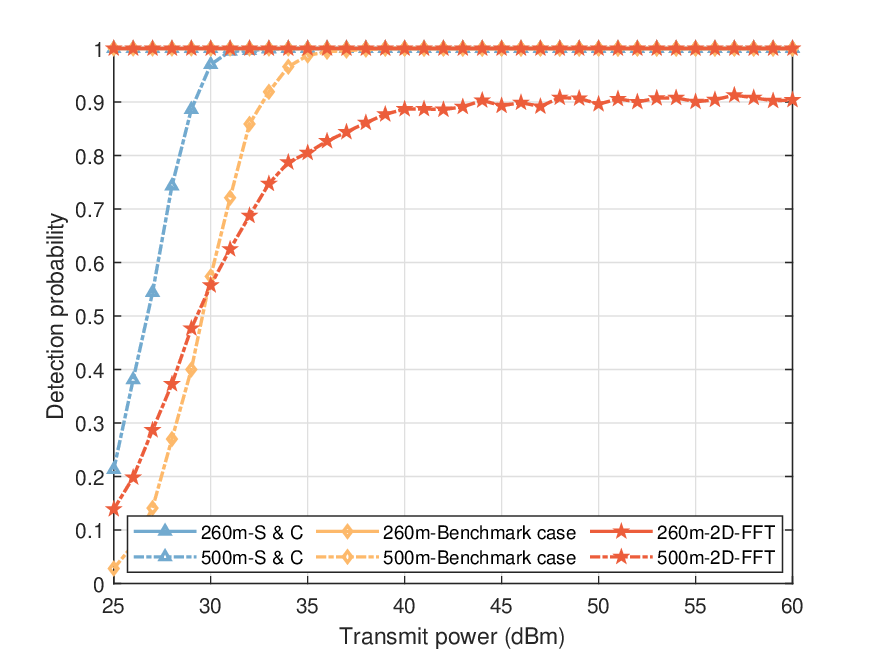}
	\caption{{The detection probability versus the transmit power under different methods. The curves corresponding to the $260$ m target overlap.}}
	\label{Pd under different power}
\end{figure}
Fig. \ref{Pd under different power} shows the detection probability versus the transmit power for the S \& C, benchmark case, and traditional 2D-FFT methods. The CFAR detector is applied for target detection, where the false alarm rate is set as $10^{-11}$. $1000$ Monte Carlo simulations are performed to calculate the detection probability. In the real physical world, the location of the target is unknown before sensing, thus the length of compensation samples can be chosen based on the prior information of the transmit beam. In this context, the length of compensation samples is determined based on the optimal length for the target at $500$ m. In Fig. \ref{Pd under different power}, the detection probability of the target at $260$ m is as high as $100\%$ and the simulated curves for the three methods overlap. This is because the power of the useful echo signal is large and the detection of target suffers less from ISI and ICI. Although the length of compensation samples is based on the target at $500$ m, it has little effect on the detection probability of the close-range target. For the target at $500$ m, the detection probability of the traditional 2D-FFT method is only $90\%$ because of the limited SINR. By contrast, the proposed S \& C and benchmark case methods greatly improve the detection probability. The detection probability of the benchmark case method is lower than the 2D-FFT method at low transmit power due to the amplified noise, as analyzed in Fig. \ref{The SINR of RDM versus the transmit power}.

\begin{figure}[t]
\centering
\includegraphics[scale=0.64]{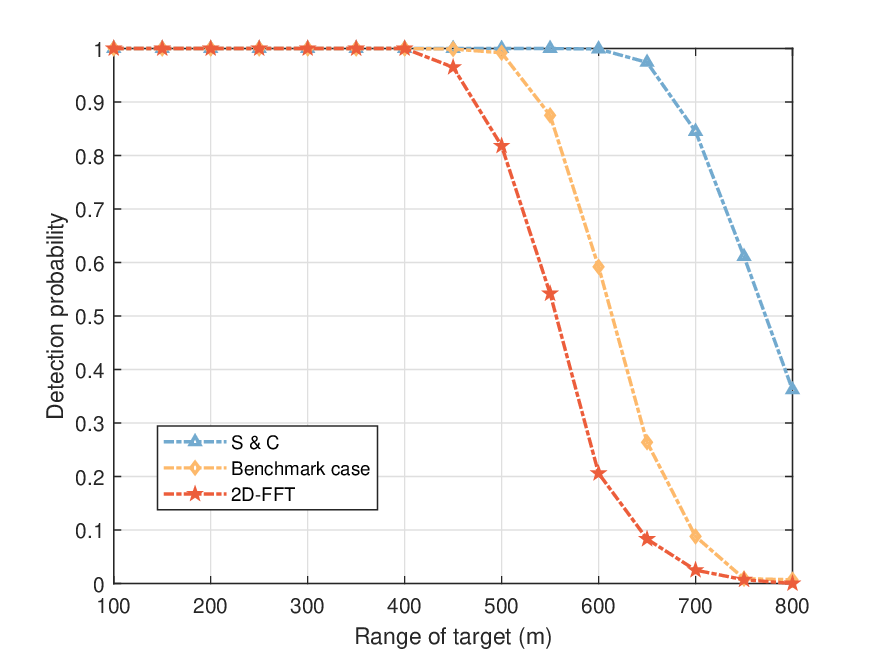}
\caption{{The detection probability versus the range of target, where the target at each range is interfered by the nearest target within the same beam. }}
\label{Pd under different range}
\end{figure}

\begin{figure}[t]
\centering
\includegraphics[scale=0.64]{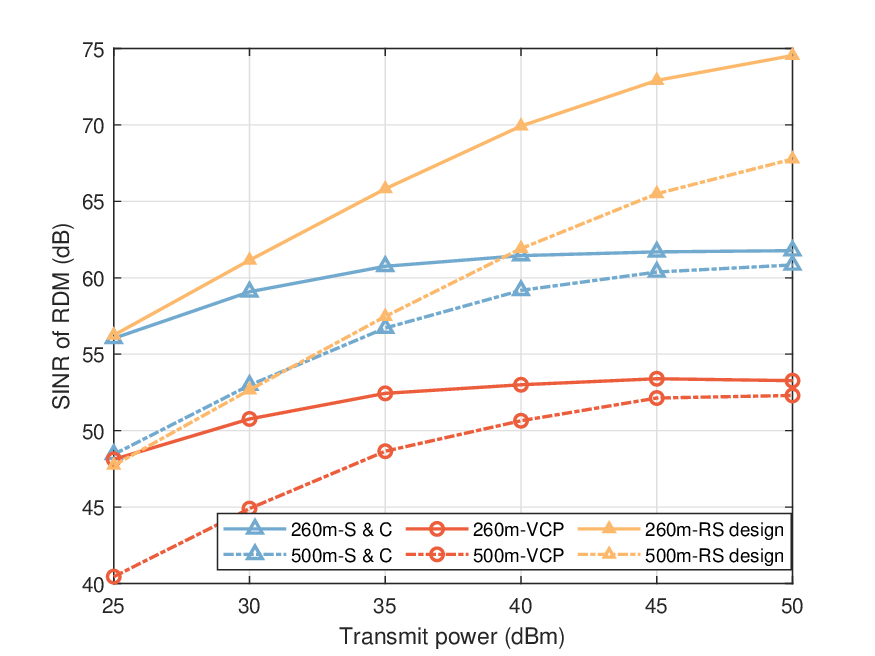}
\caption{The comparison of proposed S \& C method, VCP method \cite{wu VCP}, and RS design method \cite{Tang SCP}. }
\label{compared to VCP}
\end{figure}

Fig. \ref{Pd under different range} shows the detection probability versus the range of target, where the target at each range is interfered with the closest target within the same beam. The proposed benchmark case and S \& C methods improve the detection probability of the long-range target compared with the 2D-FFT method. For the S \& C method, the detection probability remains about $96\%$ for the target at $650$ m. In comparison, the detection probabilities of the benchmark case method and 2D-FFT method are only $25\%$ and $9\%$, respectively. When the range of the target exceeds $800$ m, both the benchmark case and 2D-FFT methods have poor performance in detection probability. By contrast, the S \& C method exhibits a significant advantage, with a detection probability of approximately $40\%$. Therefore, our proposed method significantly enhances the detection capability for long-range targets. 

Fig. \ref{compared to VCP} compares the SINR of RDM between the proposed S \& C method, the VCP method in \cite{wu VCP}, and the RS design method in \cite{Tang SCP}. The length of compensation samples for the S \& C method is selected based on the target at $500$ m. To maintain the same computational complexity, the length of the re-segmented sub-block for the VCP method is $N_c$, with adjacent sub-blocks having no overlap. Moreover, the length of VCP is the number of offset samples related to the target at $500$ m. For RS design method, the number of RSs is equal to the number of OFDM symbols in S \& C and VCP methods, and the power of each subcarrier is doubled to maintain the same transmit power. As shown in Fig. \ref{compared to VCP}, the S \& C method outperforms the VCP method for BS monostatic sensing. 
This advantage stems from that the re-segmented sub-blocks in the VCP method disrupt the original structure of OFDM, resulting in a distribution approaching the Gaussian distribution for recovered frequency domain symbols. Consequently, noise is significantly amplified during the element-wise-division process in the 2D-FFT method, deteriorating the sensing performance. 
As the second half of each reflected OFDM symbol is ISI and ICI-free, the RS design method has the best performance. However, half of the subcarriers are required to be silent, sacrificing spectrum utilization. Although the S \& C method cannot completely eliminate the ISI and ICI, it is suitable for the existing communication signal format without changing the structure of the reference signal and achieves a compromise between spectrum utilization and sensing performance.

\section{Conclusion}
\label{sec: Conclusion}
This paper focuses on the monostatic sensing for MIMO-OFDM ISAC BS. Given the wide beam, the interference from the close-range targets may drown the long-range target. To address this problem, a MUSIC and LS estimation-based spatial signal separation method is proposed to separate the echo signals reflected from different targets. Moreover, a coherent compensation-based sensing signal processing method at the receiver is proposed to enhance the SINR of RDM. The theoretical derivations for SINR of RDM obtained by the proposed method are analyzed. Simulation results verify the performance improvement of the proposed method in the SINR of RDM and detection probability for long-range targets compared to the traditional 2D-FFT method. 
\section*{Appendix}
\begin{figure*}
	\begin{align}
		{P_s} = \mathbb{E}\left[ {I_s\left( p \right){I_s^*}\left( p \right)} \right] &= {\left| {\frac{1}{{{N_c}}}{{\tilde \alpha }_u}{{\bf a}^T}\left( {{\theta _u}} \right){{\bf w}_{tx}}} \right|^2}\mathbb{E}\left[ {\sum\limits_{k = 0}^{{N_c} - 1} {\sum\limits_{i = 0}^{{N_e} - 1} {\sum\limits_{l = 0}^{{N_c} - 1} {\sum\limits_{r = 0}^{{N_e} - 1} {{S_{n - 1,k}}S_{n - 1,l}^*{e^{\frac{{j2\pi \left( {k - l} \right)\left( {{T_{cp}} - {\tau _u}} \right)}}{{{N_c}{T_s}}}}}{e^{\frac{{j2\pi \left( {k - p} \right)i}}{{{N_c}}}}}{e^{\frac{{ - j2\pi \left( {l - p} \right)r}}{{{N_c}}}}}} } } } } \right] \notag \\
		&= {\left| {\frac{1}{{{N_c}}}{{\tilde \alpha }_u}{{\bf a}^T}\left( {{\theta _u}} \right){{\bf w}_{tx}}} \right|^2}\sum\limits_{k = 0}^{{N_c} - 1} {\sum\limits_{i = 0}^{{N_e} - 1} {\sum\limits_{l = 0}^{{N_c} - 1} {\sum\limits_{r = 0}^{{N_e} - 1} {\mathbb{E}\left[ {{S_{n - 1,k}}S_{n - 1,l}^*} \right]{e^{\frac{{j2\pi \left( {k - l} \right)\left( {{T_{cp}} - {\tau _u}} \right)}}{{{N_c}{T_s}}}}}{e^{\frac{{j2\pi \left( {k - p} \right)i}}{{{N_c}}}}}{e^{\frac{{ - j2\pi \left( {l - p} \right)r}}{{{N_c}}}}}} } } } \notag\\
		& \mathop  = \limits^{\left( { a} \right)} {\left| {\frac{1}{{{N_c}}}{{\tilde \alpha }_u}{{\bf a}^T}\left( {{\theta _u}} \right){{\bf w}_{tx}}} \right|^2}\mathbb{E}\left[ {{{\left| {{S_{n - 1,k}}} \right|}^2}} \right]\sum\limits_{k = 0}^{{N_c} - 1} {\sum\limits_{i = 0}^{{N_e} - 1} {\sum\limits_{r = 0}^{{N_e} - 1} {{e^{\frac{{j2\pi \left( {k - p} \right)\left( {i - r} \right)}}{{{N_c}}}}}} } } 
		\label{derive the power of ISI}.\tag{68}
	\end{align}
\end{figure*}
\begin{figure*}
	\begin{align}
		{P_c} = \mathbb{E}\left[ {I_c\left( p \right){I_c^*}\left( p \right)} \right]
		&= {\left| {\frac{1}{{{N_c}}}{{\tilde \alpha }_u}{{\bf a}^T}\left( {{\theta _u}} \right){w_{tx}}} \right|^2}\mathbb{E}\left[ {\sum\limits_{k = 0,k \ne p}^{{N_c} - 1} {\sum\limits_{i = {N_e}}^{{N_c} - 1} {\sum\limits_{l = 0,l \ne p}^{{N_c} - 1} {\sum\limits_{r = {N_e}}^{{N_c} - 1} {{S_{n,k}}S_{n,l}^*{e^{\frac{{ - j2\pi k{\tau _u}}}{{{N_c}{T_s}}}}}{e^{\frac{{j2\pi \left( {k - p} \right)i}}{{{N_c}}}}}{e^{\frac{{ - j2\pi l{\tau _u}}}{{{N_c}{T_s}}}}}{e^{\frac{{ - j2\pi \left( {l - p} \right)r}}{{{N_c}}}}}} } } } } \right] \notag \\
		&= {\left| {\frac{1}{{{N_c}}}{{\tilde \alpha }_u}{{\bf a}^T}\left( {{\theta _u}} \right){w_{tx}}} \right|^2}\mathbb{E}\left[ {{{\left| {{S_{n,k}}} \right|}^2}} \right]\sum\limits_{k = 0,k \ne p}^{{N_c} - 1} {\sum\limits_{i = {N_e}}^{{N_c} - 1} {\sum\limits_{r = {N_e}}^{{N_c} - 1} {{e^{\frac{{j2\pi \left( {k - p} \right)\left( {i - r} \right)}}{{{N_c}}}}}} } }.
		\label{derive the power of ICI} \tag{71}
	\end{align}
{\noindent} \rule[-10pt]{18cm}{0.05em}
\end{figure*}
\subsection{Proof of Theorem \ref{theorem 1}}
The angles of targets estimated using the MUSIC method are $\boldsymbol{\theta} $. Thus, the value of ${{\bf B}\left( \boldsymbol{\theta}  \right)}$ can be obtained. The LS estimation of ${\bf x}\left( t \right)$ can be expressed as
\begin{equation}
	\mathop {\min }\limits_{{\bf \hat x}\left( t \right)} {\left\| {{\bf y}\left( t \right) - {\bf B}\left( \boldsymbol{\theta}  \right){\bf \hat x}\left( t \right)} \right\|_2^2}.
	\tag{59}\label{LS estimation 1}
\end{equation}
For ease of expression, \eqref{LS estimation 1} is written as 
\begin{equation}
	\mathop {\min }\limits_{\bf \hat x} {\left\| {{\bf y} - {\bf B\hat x}} \right\|_2^2}.\tag{60}
\end{equation}
Let $f\left( {\bf \hat x} \right) = \left\| {{\bf y} - {\bf B\hat x}} \right\|_2^2$, then we have
\begin{equation}
	f\left( {\bf \hat x} \right) = tr\left( {{\bf y}{{\bf y}^H} - {\bf y}{{\bf \hat x}^H}{{\bf B}^H} - {\bf B}{\bf \hat x}{{\bf y}^H} + {\bf B}{\bf \hat x}{{\bf \hat x}^H}{{\bf B}^H}} \right).\tag{61}
\end{equation}
Deriving $f\left( {\bf \hat x} \right)$ with respect to $ {\bf \hat x} $, we obtain
\begin{equation}
	\frac{\partial }{{\partial {\bf \hat x}}}f\left( {\bf \hat x} \right) =  - {\left( {{{\bf y}^H}{\bf B}} \right)^T} + {\left( {{{\bf B}^H}{\bf B}} \right)^T}\bf \hat x^*.\tag{62}
\end{equation}
By letting  $\frac{\partial }{{\partial {\bf \hat x}}}f\left( {\bf \hat x} \right) = {\bf 0}$, the LS estimation of ${\bf \hat x}$ is given by 
\begin{equation}
	{\bf \hat x} = {\left( {{{\bf B}^H}{\bf B}} \right)^{ - 1}}{{\bf B}^H}{\bf y}.\tag{63}
	\label{LS estimation}
\end{equation}
Thus, the separated signals can be expressed as \eqref{signal seperation}.

\subsection{Proof of Theorem \ref{theorem 2}}
Sampling ${\bf n}\left( t \right)$ at $t=t_i$, we obtain
\begin{equation}
	{\bf n}\left( {{t_i}} \right){\rm{ = }}{\left[ {{n_0}\left( {{t_i}} \right),{n_1}\left( {{t_i}} \right), \cdots ,{n_{{N_r} - 1}}\left( {{t_i}} \right)} \right]^T},\tag{64}
\end{equation}
where the noise ${n_r}\left( {{t_i}} \right)$, $r=0,1,\cdots,N_r-1$, on different antenna elements follows i.i.d CSCG distribution, i.e., ${n_r}\left( {{t_i}} \right)\sim {\cal C}{\cal N}\left( {0,{\sigma ^2}} \right)$. According to \eqref{signal seperation}, we gain
\begin{equation}
	{\bar n_u}\left( {{t_i}} \right) = {{\bf b}^T_u}{\bf n}\left( {{t_i}} \right) 
	 = \sum\limits_{r = 0}^{{N_r} - 1} {{b_{u,r}}{n_r}\left( {{t_i}} \right)},
	 \label{seperated noise}\tag{65}
\end{equation}
where ${{\bf b}^T_u} = \left[ {{b_{u,0}},{b_{u,1}}, \cdots ,{b_{u,{N_r} - 1}}} \right]$ is the $u$th row of Moore-Penrose inverse matrix ${\left[ {{\bf B}\left( {\boldsymbol{\theta}}  \right)} \right]^\dag }$. 
Then, the expectation and variance of ${\bar n_u}\left( {{t_i}} \right)$ can be, respectively, expressed as
\begin{equation}
	\mathbb{E} \{{\bar n_u}\left( {{t_i}}\right)\}=0,\tag{66}
\end{equation}
\begin{equation}
	\mathbb{V} \{{\bar n_u}\left( {{t_i}}\right)\}={\lambda _u}{\sigma ^2},\tag{67}
\end{equation}
where $\lambda_u=\sum\limits_{r = 0}^{{N_r} - 1} {{{\left| {{b_{u,r}}} \right|}^2}}$ is the $k$th element on the diagonal of ${\left[ {{\bf B}\left( \boldsymbol{\theta}  \right)} \right]^\dag }{[ {{{[ {{\bf B}\left(  \boldsymbol{\theta}  \right)} ]}^\dag }} ]^H}$. Thus, we have ${\bar n_u}\left( {{t_i}} \right) \sim {\cal C}{\cal N}\left( {0,{\lambda _u}{\sigma ^2}} \right)$. The proof is completed.

\subsection{Proof of \eqref{power of ISI signal before} and \eqref{power of ICI signal before}}

According to \eqref{frequency domain symbol}, the expectation of $I_s(p)$ is zero. Thus, the power of ISI is the variance of $I_s(p)$, as shown in \eqref{derive the power of ISI} at the top of next page, where $(a)$  applies the independence between transmitted symbols. Only when $k=l$ and $i \ne r$, ${\mathbb{E}[ {{S_{n - 1,k}}S_{n - 1,l}^*}]}$ is not zero and summing $\exp ( {\frac{{j2\pi \left( {k - p} \right)\left( {i - r} \right)}}{{{N_c}}}})$ from $k=0$ to $k=N_c-1$ equals zero. Moreover, we obtain
\begin{align}
		&\sum\limits_{k = 0}^{{N_c} - 1} {\sum\limits_{i = 0}^{{N_e} - 1} {\sum\limits_{r = 0}^{{N_e} - 1} {{e^{\frac{{j2\pi \left( {k - p} \right)\left( {i - r} \right)}}{{{N_c}}}}}} } }  \notag \\
		&= \sum\limits_{k = 0}^{{N_c} - 1} {\sum\limits_{i = 0,i \ne r}^{{N_e} - 1} {\sum\limits_{r = 0,r \ne i}^{{N_e} - 1} {{e^{\frac{{j2\pi \left( {k - p} \right)\left( {i - r} \right)}}{{{N_c}}}}}}  + \sum\limits_{k = 0}^{{N_c} - 1} {\sum\limits_{i = 0}^{{N_e} - 1} 1 } } }  \notag \\
		&= {N_c}{N_e}.\tag{69}
\end{align}
Therefore, the power of ISI can be expressed as
\begin{equation}
	P_s={{{\tilde N}_e}}{\left| {{{\tilde \alpha }_u}{{\bf a}^T}\left( {{\theta _u}} \right){{\bf w}_{tx}}} \right|^2}\mathbb{E}\left[ {{{\left| {{S_{n - 1,k}}} \right|}^2}} \right].\tag{70}
\end{equation}

Similarly, the power of ICI is the variance of $I_c(p)$, as shown in \eqref{derive the power of ICI} at the top of next page. Furthermore, we gain
\begin{align}
		\sum\limits_{k = 0,k \ne p}^{{N_c} - 1} &{\sum\limits_{i = {N_e}}^{{N_c} - 1} {\sum\limits_{r = {N_e}}^{{N_c} - 1} {{e^{\frac{{j2\pi \left( {k - p} \right)\left( {i - r} \right)}}{{{N_c}}}}}} } } \tag{72}\\
		&= \sum\limits_{k = 0}^{{N_c} - 1} {\sum\limits_{i = {N_e},i \ne r}^{{N_c} - 1} {\sum\limits_{r = {N_e},r \ne i}^{{N_c} - 1} {{e^{\frac{{j2\pi \left( {k - p} \right)\left( {i - r} \right)}}{{{N_c}}}}}} } } \notag\\
		&- \sum\limits_{i = {N_e},i \ne r}^{{N_c} - 1} {\sum\limits_{r = {N_e},r \ne i}^{{N_c} - 1} 1 }  + \sum\limits_{k = 0,k \ne p}^{{N_c} - 1} {\sum\limits_{i = {N_e}}^{{N_c} - 1} 1 } \notag\\
		&= {N_e}\left( {{N_c} - {N_e}} \right).\notag
\end{align}
Thus, the power of ICI is given by
\begin{equation}
	P_c={{{\tilde N}_e}}( {1 - {{{\tilde N}_e}}}){\left| {{{\tilde \alpha }_u}{{\bf a}^T}\left( {{\theta _u}} \right){{\bf w}_{tx}}} \right|^2}\mathbb{E}\left[ {{{\left| {{S_{n,k}}} \right|}^2}} \right].\tag{73}
\end{equation}

\subsection{Proof of \eqref{power of noise after coherent compensation}}
From \eqref{frequency symbol after coherent compensation} we can obtain
\begin{equation}
	{\tilde N_u}\left( p \right) = \frac{1}{{\sqrt {{N_c}} }}\sum\limits_{i = 0}^{{N_c} - 1} {{{\tilde n}_u}\left( i \right){e^{j2\pi \frac{{ip}}{{{N_c}}}}}}.\tag{74}
\end{equation}
Since the expectation of ${\tilde N_u}\left( p \right)$ is zero, the power of ${\tilde N_u}\left( p \right)$ is the variance of ${\tilde N_u}\left( p \right)$, which can be expressed as
\begin{align}
	{{\tilde P}_n} &= \mathbb{E}\left[ {{{\tilde N}_u}\left( p \right)\tilde N_u^*\left( p \right)} \right]\notag \\
	&= \mathbb{E}\left[ {\frac{1}{{{N_c}}}\sum\limits_{i = 0}^{{N_c} - 1} {\sum\limits_{k = 0}^{{N_c} - 1} {{{\tilde n}_u}\left( i \right)\tilde n_u^*\left( k \right)} {e^{j2\pi \frac{{\left( {i - k} \right)p}}{{{N_c}}}}}} } \right]\notag \\
	&\mathop  = \limits^{\left( a \right)} \frac{1}{{{N_c}}}\sum\limits_{i = 0}^{{N_c} - 1} {E\left[ {{{\left| {{{\tilde n}_u}\left( i \right)} \right|}^2}} \right]}\notag \\
	&\mathop  = \limits^{\left( b \right)} ( {1 + {{\tilde N}_a}} ){\lambda _u}{\sigma ^2}.\tag{75}
\end{align}
where $(a)$ implies the independence between ${{{\tilde n}_u}\left( i \right)}$ and $(b)$ applies \eqref{noise in time domain after coherent compensation}.

\small
\bibliographystyle{ieeetr}

\end{document}